\documentclass[12pt,titlepage]{article}

\usepackage{geometry}
\usepackage{mathtools,amssymb,amsfonts,amsthm,amsmath,amscd}

\usepackage{enumitem}
\usepackage{bm}

\usepackage{pifont}
\usepackage{fontawesome}
\usepackage{wasysym}

\usepackage[utf8]{inputenc} %Décochez si vous utilisez un encodage utf-8
\usepackage[T1]{fontenc}

\usepackage[british]{babel}

\usepackage[authoryear]{natbib}

\usepackage{graphicx}

\usepackage{booktabs,lipsum} 
\usepackage{multirow}

\usepackage{lineno}
\usepackage{ulem}

\usepackage{titlesec}
\titleformat{\chapter}[display]
  {\normalfont\huge\bfseries}
  {\chaptertitlename\ \thechapter}{20pt}{\huge}

\usepackage{hyperref}
\hypersetup{
    colorlinks=true,   % false: liens encadrés; true: liens colorés
    linkcolor=black,   % couleur des liens (ou bordures) internes
    citecolor=magenta, % couleur des liens (ou bordures) vers bilbio
    filecolor=green,   % couleur des liens (ou bordures) vers fichiers
    urlcolor=red       % couleur des liens (ou bordures) url
}

\usepackage{color}
\usepackage[all]{xy}
\usepackage[colorinlistoftodos]{todonotes}
\usepackage{tikz-cd}
\usepackage{systeme}

\usepackage{url}
\usepackage{IEEEtrantools}

\usepackage[T1]{fontenc}

\usepackage[cal=pxtx,frak=euler,scr=boondox,bb= pazo]{mathalfa}

\usepackage{orcidlink}

\newtheorem{Lemma}{Lemma}[section] %numbering by section
\newtheorem{Remark}{Remark}[section] %numbering by section
\newtheorem{Theorem}{Theorem}[section] %numbering by section
\newtheorem*{Assumptions}{Assumptions} %numbering by section

\DeclareMathOperator*{\I}{\text{I}}
\DeclareMathOperator*{\Proba}{\text{P}}

\newcommand{\C}{{\cal C}}

\newcommand{\bfor}{\begin{eqnarray*}}
\newcommand{\efor}{\end{eqnarray*}}

\begin{document}

%\linenumbers

%%%%%%%%%%%%%%%%%%%%%%%%%%%%%%%%%%%%%%%%%%%%%%%%%%%%%%%%
% TITLE PAGE
%%%%%%%%%%%%%%%%%%%%%%%%%%%%%%%%%%%%%%%%%%%%%%%%%%%%%%%%

\title{\bf Copula based dependent censoring in cure models with covariates}
\author{Morine Delhelle$^{\,a\,\orcidlink{0009-0008-4032-4233}}$, Anouar El Ghouch$^{\,a}$ and Ingrid Van Keilegom$^{\,a,b\,\orcidlink{0000-0001-8827-7642}}$ \footnote{morine.delhelle@uclouvain.be, anouar.elghouch@uclouvain.be, ingrid.vankeilegom@kuleuven.be.}\\[5mm]
$^a$ Institute of Statistics, Biostatistics and Actuarial Sciences (ISBA), UCLouvain, Belgium\\[3mm]
$^b$ Research Centre for Operations Research and Statistics (ORSTAT), KU Leuven, Belgium}
\date{}
\maketitle

\begin{abstract}
In survival analysis, the time-to-event variable $T$ is frequently subject to right censoring. Individuals may withdraw from the study for various reasons, or may not experience the event of interest before the end of follow-up. In this paper, we distinguish between two types of censoring: a potentially dependent censoring time $C$, which may be stochastically related to $T$, and an independent administrative censoring time $A$. In addition, the data may exhibit a cure fraction, meaning that some individuals will never experience the event. We build upon a recent work about a fully parametric mixture cure model, which accounts for dependent censoring through copulas. The proposed extension incorporates administrative censoring and allows covariates to affect all model parameters. This framework enables a more accurate modelling of the dependence between survival and censoring times while providing greater flexibility through covariate effects, leading to more individualised estimation of the cure fraction, the dependence structure, and other clinically relevant quantities. Moreover, the presence of covariates allows for weaker identification conditions.
\\[5mm]
\end{abstract}

%%%%%%%%%%%%%%%%%%%%%%%%%%%%%%%%%%%%%%%%%%%%%%%%%%%%%%%%%%%%%%%%%%%%%
\section{Introduction} \label{sect1Intro}
%%%%%%%%%%%%%%%%%%%%%%%%%%%%%%%%%%%%%%%%%%%%%%%%%%%%%%%%%%%%%%%%%%%%%

Survival analysis is a cornerstone of biomedical research, where the objective is to model the time until the occurrence of an event of interest such as death, disease recurrence, or treatment failure. Beyond medicine, survival analysis also plays an important role in many other fields including demography (e.g., time to first childbirth) and engineering (e.g., time until machine failure).

In practice, the time-to-event variable $T$ is frequently only partially observed because of \textit{right censoring}. This occurs when individuals withdraw from a study or remain event-free at the end of follow-up. Traditional approaches, such as the Cox proportional hazards model \cite{Cox}, generally rely on the assumption of independent censoring, meaning that the censoring time carries no information about the unobserved event time. However, this assumption is often unrealistic in practice. For example, in oncology trials, patients may withdraw from a study following a substantial improvement in their health status, which is directly related to the event of interest (e.g., cancer progression). Ignoring such dependent censoring may result in biased estimates of survival probabilities, hazard ratios, and other clinically relevant quantities (\cite{LiAl}, \cite{RondeauAl2011}, \cite{EmuraChen2018}). Consequently, a variety of methods have been proposed to model the stochastic dependence between survival and censoring times. These approaches may be broadly divided into two main categories. The first comprises fully specified copula-based models, such as \cite{RivestWells} and \cite{EmotoMatthews}. The second includes approaches in which the dependence structure is estimated from the data, such as the parametric copula model with an unknown dependence parameter proposed by \cite{CzadoVanKeilegom}, and the semiparametric framework of \cite{DeresaVK_2024}, which combines an unknown parametric copula with a Cox model. An overview of recent advances in methods using copulas is provided by \cite{CrommenAl_2025}.

As explained above, right censoring time is often stochastically dependent on the survival time we are interested in; however, this is not always the case. The most common example is that of \textit{administrative censoring}, which arises from study termination and is considered as a non-informative censoring mechanism. Therefore, it is useful to be able to take these two types of censoring into account separately when modelling.

An additional layer of complexity arises when the data exhibit a \textit{cure fraction}, meaning that a proportion of individuals will never experience the event of interest (e.g., patients who never relapse following treatment). \textit{Cure models} describe the population survival function through the \textit{incidence}, defined as the probability of being uncured, and the \textit{latency}, defined as the survival function of uncured individuals. Comprehensive discussions of cure models can be found in \cite{MallerZhou_1996}, \cite{PengYu_2022}, and Chapter 4 of the monograph by \cite{Legrand2021}. The seminal \textit{mixture cure model} (\cite{BerksonGage}) is a formalisation of this framework that has been extensively studied (e.g., \cite{PengTaylor_2014} and \cite{AmicoVK}) and consists in distinguishing between a subpopulation of susceptible individuals (who will experience the event) and a subpopulation of non-susceptibles (who are event free).

Despite their practical relevance, the joint treatment of dependent censoring and cure fractions remains relatively underdeveloped and existing approaches often rely on restrictive assumptions. For example, the parametric model of \cite{DeresaVKAntonio} accommodates dependent censoring and covariate effects but does not account for a cure fraction, \cite{LiAl} assumes the copula to be known, \cite{Bernhardt} considers a known cure threshold, \cite{SchneiderDemarquidFC} is limited to clustered data, \cite{OthusAl} requires the presence of covariates and for some of them the model identifiability is not established (e.g., \cite{Bernhardt}, \cite{LiuAl}). More recently, \cite{DelhelleVK} proposed a copula-based framework for modelling dependent censoring in cure models. However, their approach neither incorporated covariates, which are essential in many practical applications, nor accounted for administrative censoring. Furthermore, their framework relied on the assumption that the latency distribution has bounded support.

In this paper, we extend the model proposed by \cite{DelhelleVK}. The proposed framework explicitly distinguishes administrative censoring from potentially dependent censoring. In addition, it allows \textit{covariates} to affect all model parameters. This extension provides a more flexible and realistic framework for analysing survival data, enabling more accurate modelling of the dependence between survival and censoring times while allowing individualised estimation of the cure fraction, the dependence structure, and other clinically relevant quantities.

The remainder of the paper is organised as follows. Section \ref{sect2Mod} presents the proposed model. Section \ref{sect3IdentParamEst} establishes its identifiability and describes the parameter estimation procedure. The finite-sample performance and robustness of the method are evaluated through simulation studies in Section \ref{sect4Simus}. Section \ref{sect5Data} illustrates the practical usefulness of the proposed approach through the analysis of a real dataset. Finally, Section \ref{sect6Conclu} concludes the paper.

%%%%%%%%%%%%%%%%%%%%%%%%%%%%%%%%%%%%%%%%%%%%%%%%%%%%%%%%%%%%%%%%%%%%%
\section{The model} \label{sect2Mod}
%%%%%%%%%%%%%%%%%%%%%%%%%%%%%%%%%%%%%%%%%%%%%%%%%%%%%%%%%%%%%%%%%%%%%

Let $T$, $C$ and $A$ be three non-negative continuous random variables, representing the survival time, the dependent censoring time, and the administrative censoring times, respectively. In the context of right censored data, we observe $(Y, \Delta_T, \Delta_C)$, with $Y=\min\left(T, C, A\right)$ and the indicators $\Delta_T=\I\left(Y=T\right)$ and $\Delta_C=\I\left(Y=C\right)$. To each individual a $m-$dimensional vector of covariates $X$ is associated. In the framework considered, a cure fraction is assumed, meaning that some individuals will never experience the event of interest, and their survival time is therefore considered infinite. The positive probability that the event does not occur, called the cure rate, is denoted by $1-\pi(x)$ where $\pi(x)=\Proba\left(T<\infty|X=x\right)$ is the incidence of the model. The random variable corresponding to the survival time of uncured individuals is denoted by $U$. %Note that if we define $B=\I\left(T|X < \infty\right)$, then $T|X =U|X B+\infty\left(1-B\right)$, where $U|X$ is the survival time for the uncured individuals. 

Let $S_{T|X}(t|x)=\Proba\left(T>t|X=x\right)$ denote the improper conditional survival function of $T|X$. According to the mixture cure model, we have
\begin{equation}
S_{T|X}(t|x)=1-\pi(x)+\pi(x)S_{U|X}(t|x), \label{eq:modelmixture}
\end{equation}
where $S_{U|X}(t|x)=\Proba\left(T>t|T<\infty, X=x\right)$ is the proper survival function for uncured individuals, referred to as the latency component of the model. With the cumulative distribution functions $F_{T|X}= 1-S_{T|X}$ and $F_{U|X}=1-S_{U|X}$, the model can also be written as $F_{T|X}(t|x)=\pi(x)F_{U|X}(t|x)$.\\
An important component of the model, which distinguishes it from standard approaches in the literature, is the stochastic dependence between $T$ and $C$.
As in \cite{DelhelleVK}, this dependency is modelled by writing the bivariate distribution of $(T,C)|X$ as a copula of the marginal distributions:
\begin{equation}
F_{T,C|X}(t,c|x)=\Proba\left(T\leq t, C\leq c|X=x\right)=\C\left(F_{T|X}(t|x), F_{C|X}(c|x)\big|x\right). \label{eq:modelcopula}
\end{equation}
The copula $\C\left(\cdot, \cdot\big|x\right)$ is a bivariate distribution defined on $[0, 1]^2$ with uniform margins. Further details on copulas can be found in the monograph by \cite{Nelsen}.
Finally, the administrative censoring time $A$, whose support is $[0, \infty)$, is assumed to be independent of both ($T$, $C$) and the covariates $X$, reflecting its non-informative nature. Its density and distribution functions are denoted by $f_A$ and $F_A$, respectively.

The proposed model is fully parametric. The conditional distribution $F_{U|X}(\cdot|x)$, which depends on the covariates, is specified parametrically using a GLM-type formulation for each of its parameters:
$$F_{U|X}^{\beta_U}(u|x)=H_{U}\left(u; \theta_U(x)\right),$$
where $H_{U}$ denotes a known distribution function (e.g., Weibull or log-normal) and $\theta_U(x)=g\left(\beta_{U}^{\top}x\right)$. The function $g$ is a known link function, required to be one-to-one (see Subsection \ref{subsect3Ident}) and chosen to satisfy the constraint implied by the parameter space of $\theta_U$. Note that, depending on the family of the distribution function, there may, of course, be more than one parameter indexing $H$. Note that, for notational convenience, the vector $X$ includes the constant term $1$ in addition to the covariates.

As an illustration, consider $X=(1,X_1, X_2)^\top$ and assume that $U$ follows a Weibull distribution. Since both the shape and scale parameters are strictly positive, exponential link functions can be used. In this case, the shape parameter is given by $\theta_{U1}(x)=\exp(\beta_{U_{0}}+\beta_{U_{1}}x_1+\beta_{U_{2}}x_2)$ and the scale parameter by $\theta_{U2}(x)=\exp(\beta_{U_{3}}+\beta_{U_{4}}x_1+\beta_{U_{5}}x_2)$.\\
The same modelling framework applies to the other components of the model, namely $F_{C|X}(\cdot|x)$, the copula $\C\left(\cdot, \cdot\big|x\right)$ which depends on $x$ via Kendall’s tau $\tau^{K}(x)$, and $\pi(x)$, with corresponding parameter vectors $\beta_C$, $\tau$, and $p$. For $\pi(x)$, one may use a logistic link function, but other choices are also allowed:
$$\pi_p(x)=\frac{\exp{\left(p^\top x\right)}}{1+\exp{\left(p^\top x\right)}}.$$
As for the copula, the link function depends on the admissible range of dependence. For example, when using a Frank copula, a suitable choice could be $\tau^{K}_{\tau}(x)=(\exp(\tau^\top x)-1)/(\exp(\tau^\top x)+1)$. For copulas that allow only positive or only negative dependence, alternative link functions with appropriate ranges must be employed.

The parameter vector $\beta=\left(\tau, p, \beta_U, \beta_C\right)^\top$ gathers all model coefficients, and its dimension depends on the number of covariates and the chosen parametrisation. It includes at least the intercept terms. All coefficients are estimated from the data, whereas the link functions are assumed fixed in advance.

For notational simplicity and generality, the same covariate vector $X$ is used for all model parameters. In practice, however, some covariates may not influence certain components of the model. Although this requires minor modifications, the model can be equivalently formulated using distinct covariate vectors for the copula, the marginal distributions of $U$ and $C$, and the incidence component.

Henceforth, to simplify the notation, we use the following abbreviations whenever this does not impair clarity: $F_{\beta_{U}}(t|x)$ for $F_{U|X}^{\,\beta_{U}}(t|x)$, $F_{\beta_{C}}(c|x)$ for $F_{C|X}^{\,\beta_{C}}(c|x)$, and analogous abbreviations for the corresponding densities. For what follows, we also need to define the partial derivatives of the copula:
$$h_{\tau}^{1}\left(v|u,x\right)=\frac{\partial}{\partial u}\C_{\tau}(u,v|x) \quad \text{and} \quad h_{\tau}^{2}\left(u|v,x\right)=\frac{\partial}{\partial v}\C_{\tau}(u,v|x).$$
The conditional distributions can then be expressed as
$$\Proba\left(T\leq t|C=c,X=x\right) = h_{\tau}^{2}\left(\pi_p(x)F_{\beta_{U}}(t|x)\big|F_{\beta_{C}}(c|x),x\right),$$
and
$$\Proba\left(C\leq c|T=t,X=x\right) = h_{\tau}^{1}\left(F_{\beta_{C}}(c|x)\big|\pi_p(x)F_{\beta_{U}}(t|x),x\right).$$

These expressions will be used in the next section to derive the individual likelihood contributions.

%%%%%%%%%%%%%%%%%%%%%%%%%%%%%%%%%%%%%%%%%%%%%%%%%%%%%%%%%%%%%%%%%%%%%
\section{Identifiability and parameter estimation}\label{sect3IdentParamEst}
%%%%%%%%%%%%%%%%%%%%%%%%%%%%%%%%%%%%%%%%%%%%%%%%%%%%%%%%%%%%%%%%%%%%%

This section addresses the identifiability of the model, a fundamental property that is examined in detail. Parameter estimation is also carried out using the maximum likelihood approach.

%%%%%%%%%%%%%%%%%%%%%
\subsection{Identifiability of the model}\label{subsect3Ident}
%%%%%%%%%%%%%%%%%%%%%

There are three types of contributions to the likelihood, their developments can be found in Appendix B. For an uncensored individual, we have
$$\frac{\partial}{\partial y} \Proba(Y \le y, \Delta_T=1, \Delta_C=0| X=x) = \pi_p(x)f_{\beta_{U}}(y|x) \left[1-h_{\tau}^{1}\left\{F_{\beta_{C}}(y|x)\big|\pi_p(x)F_{\beta_{U}}(y|x),x\right\}\right]\{1-F_A(y)\},$$
for a dependently censored individual, the contribution is
$$\frac{\partial}{\partial y} \Proba(Y \le y, \Delta_T=0, \Delta_C=1| X=x) = f_{\beta_{C}}(y|x) \left[1-h_{\tau}^{2}\left\{\pi_p(x)F_{\beta_{U}}(y|x)\big|F_{\beta_{C}}(y|x),x\right\}\right]\{1-F_A(y)\},$$
and for an administratively censored individual, we have
$$\frac{\partial}{\partial y} \Proba(Y \le y, \Delta_T=0, \Delta_C=0| X=x) = \overline{\C}_{\tau}\left\{1-\pi_p(x)F_{\beta_{U}}(y|x), 1-F_{\beta_{C}}(y|x)\big|x\right\}f_A(y),$$
where $\overline{\C}_{\tau}\left(u, v|x\right)=u+v-1+\C_{\tau}\left(1-u,1-v|x\right)$ denotes the survival copula.
%Additional information on this expression can be found in the articles by \cite{DelhelleVK} and \cite{DeresaVKAntonio} upon which it is based.

By combining these three components, we obtain the individual likelihood contribution for an observation $(y, \delta_T, \delta_C, x)$, which depends on the parameter vector $\beta$ (hereafter, quantities related to the distribution of $A$ are omitted, since administrative censoring is assumed to be non-informative):
\begin{align*}
    l(\beta; y, \delta_T, \delta_C, x) &= \left(\pi_p(x)f_{\beta_{U}}(y|x)\left[1-h_{\tau}^{1}\left\{F_{\beta_{C}}(y|x)|\pi_p(x)F_{\beta_{U}}(y|x),x\right\}\right]\right)^{\delta_{T}}\\
    &\hspace{0.7cm}\times \left(f_{\beta_{C}}(y|x)\left[1-h_{\tau}^{2}\left\{\pi_p(x)F_{\beta_{U}}(y|x)|F_{\beta_{C}}(y|x),x\right\}\right]\right)^{\delta_{C}}\\
    &\hspace{0.7cm}\times \left[\overline{\C}_{\tau}\left\{1-\pi_p(x)F_{\beta_{U}}(y|x), 1-F_{\beta_{C}}(y|x)|x\right\}\right]^{1-\delta_{T}-\delta_{C}}.
\end{align*}

Establishing identifiability of the model amounts to showing that, if $l(\beta; y, \delta_T, \delta_C, x) = l(\tilde\beta; y, \delta_T, \delta_C, x)$ for all $(y, \delta_T, \delta_C, x)$ in the support of $(Y, \Delta_T, \Delta_C, X)$, then necessarily $\beta=\tilde\beta$. The proof of identifiability proceeds along two main lines.

The first is an adaptation of Theorem 1 in \cite{DelhelleVK}. For technical reasons, this result assumes that the support of $U|X$ is bounded to allow certain steps in the identifiability proof, whereas the support of $C|X$ is $[0, \infty)$.

The second line of argument, which constitutes a main contribution of this work, relaxes the truncation assumption in several ways depending on the distributions under consideration. Although the presence of covariates introduces additional complexity and requires careful handling, the information they provide can be leveraged to remove the bounded support assumption in specific settings. A key contribution of this work is thus the derivation of identifiability conditions when the support of $U|X$ is unbounded, thereby broadening the applicability of the model by allowing greater flexibility in the choice of marginal distributions.

The identifiability proof relies on a set of high-level assumptions, combined differently depending on the approach considered. The assumptions are listed below and require the introduction of the cure threshold $\check{\tau}(x)=\inf\{y : F_{\beta_{U}}(y|x)=1\}$.

\begin{Assumptions}
\
\begin{enumerate}[label=A\arabic*.]
    \item $\displaystyle\lim_{y \rightarrow 0}\frac{f_{\beta_{C}}(y|x)}{f_{\tilde\beta_{C}}(y|x)}=1 \hspace{0.25cm}\forall x \Longleftrightarrow \beta_{C}=\tilde\beta_{C} \hspace{0.25cm} \mbox{and} \hspace{0.25cm} \displaystyle\lim_{y \rightarrow \infty}\frac{f_{\beta_{C}}(y|x)}{f_{\tilde\beta_{C}}(y|x)}=1 \hspace{0.25cm}\forall x \Longleftrightarrow \beta_{C}=\tilde\beta_{C}$
    \item $\displaystyle\lim_{y\rightarrow 0} h_{\tau}^{2}\left\{ \pi_p(x)F_{\beta_{U}}(y|x)|F_{\beta_{C}}(y|x),x\right\} = 0 \hspace{0.25cm}\forall x,\beta \hspace{0.25cm}\text{ or }\hspace{0.25cm} \displaystyle\lim_{y \rightarrow \infty} h_{\tau}^{2}\left\{ \pi_p(x)F_{\beta_{U}}(y|x)|F_{\beta_{C}}(y|x),x\right\} = 0 \hspace{0.25cm}\forall x,\beta$
    \item $\displaystyle\lim_{y\rightarrow 0} h_{\tau}^{1}\left\{F_{\beta_{C}}(y|x)|\pi_p(x)F_{\beta_{U}}(y|x),x\right\} = 0 \hspace{0.25cm}\forall x,\beta$
    \item $\displaystyle\lim_{y \rightarrow 0}\frac{f_{\beta_{U}}(y|x)}{f_{\tilde\beta_{U}}(y|x)}=1 \hspace{0.25cm}\forall x \Longleftrightarrow \beta_{U}=\tilde\beta_{U}$
    \item $h_{\tau}^{2}\{\pi_p(x)|F_{\beta_{C}}(y|x),x\}=h_{\tilde\tau}^{2}\{\pi_{\tilde p}(x)|F_{\beta_{C}}(y|x),x\} \hspace{0.25cm} \forall x \text{ and } \forall y > \check{\tau}(x) \Longleftrightarrow p=\tilde p \text{ and } \tau=\tilde\tau$
    \item $\displaystyle\lim_{y \rightarrow 0}\frac{\pi_p(x) f_{\beta_{U}}(y|x)}{\pi_{\tilde p}(x) f_{\tilde\beta_{U}}(y|x)}=1 \hspace{0.25cm}\forall x \Longleftrightarrow \beta_{U}=\tilde\beta_{U} \text{ and } p=\tilde p$
    \item $h_{\tau}^{2}\left\{\pi_p(x)F_{\beta_{U}}(y|x)|F_{\beta_{C}}(y|x),x\right\}=h_{\tilde\tau}^{2}\left\{\pi_{p}(x)F_{\beta_{U}}(y|x)|F_{\beta_{C}}(y|x),x\right\} \hspace{0.25cm}\forall x \text{ and } \forall y \Longleftrightarrow \tau=\tilde\tau$
\end{enumerate}
\end{Assumptions}

Assumptions A1--A5 correspond to covariate-adjusted versions of those in Theorem 1 of \cite{DelhelleVK}, while A6 is equivalent to an assumption in Theorem 1 of \cite{CzadoVanKeilegom}, rewritten in terms of incidence and covariates.

These assumptions are not all required simultaneously. For example, A5 applies only to truncated distributions and is stronger than A7. Conversely, A6 (which is stronger than A4) and A7 are used for certain non-truncated distributions. Appropriate combinations of these assumptions therefore yield identifiability results for both truncated and some non-truncated distributions of $U$, as stated in the following theorem.

\begin{Theorem}[Identification]\label{IdentMod}
    Assume that the link functions are one-to-one and the components of $X$ are linearly independent. If assumptions A1 to A3 are verified and, either A4 and A5, or A6 and A7 are also verified, then the model defined by (\ref{eq:modelmixture})-(\ref{eq:modelcopula}) is identified.
\end{Theorem}
The proof of this theorem is provided in Appendix B.

In the remainder of this paper, we focus on some Archimedean copulas as well as on the Gaussian copula, which belongs to the class of elliptical copulas. It should be emphasized, however, that the results presented are not restricted to these specific families, others may also be employed provided that the above assumptions are satisfied. Moreover, the developments presented in this document may serve as a basis for establishing analogous results in other settings.

We recall the definition of Archimedean copulas:
$$\C_{\tau}(u, v|x)=\varphi^{[-1]}_{\tau}\left(\varphi_{\tau}(u|x)+\varphi_{\tau}(v|x)\big|x\right) \quad \forall u, v \in [0, 1],$$
where $\varphi_{\tau}$ is a \textit{generator}, i.e., a continuous, convex, and strictly decreasing function from $[0, 1]$ to $[0, \infty]$ such that $\varphi_{\tau}(1)=0$, and $\varphi_{\tau}^{[-1]}$ denotes its pseudo-inverse.

The Gaussian copula is defined, for $u, v \in [0, 1]$, by 
$$\C_{\tau}(u, v|x)=\Phi_{\tau}\left(\Phi^{-1}(u), \Phi^{-1}(v)\big|x\right),$$
where $\Phi$ denotes the cumulative distribution function of the standard normal distribution and $\Phi_{\tau}(\cdot, \cdot|x)$ is the joint cumulative distribution function of a bivariate standard normal random vector with a correlation matrix determined by $\tau^{K}_{\tau}(x)$.

Combining Theorem 2 of \cite{DelhelleVK} with Lemma \ref{IdentAlpha} in Appendix B shows that assumption A1 holds for a broad class of marginal distributions (e.g., log-normal, log-logistic, Weibull), as does assumption A4 in the case of truncated distributions. In addition, several copula families (e.g., Frank, Joe, Clayton(90)\footnote{Clayton($\alpha$) denotes the Clayton copula rotated by $\alpha$ degrees.}) satisfy assumptions A2, A3, and A5, typically independently of the marginal distributions, with only mild additional conditions required for the Gumbel and Gaussian copulas.\\
Assumptions A6 and A7 enable alternative identifiability arguments in settings where the distribution of $U$ is non-truncated, which constitute a central contribution of this work. This extension broadens the applicability of the model and provides increased flexibility for data analysis. The following theorem specifies the conditions under which these two assumptions are satisfied.

\begin{Theorem}[Assumptions A6 and A7]\label{IdentA6A7}
    Assume that the link functions are one-to-one and the components of $X$ are linearly independent.
    \begin{enumerate}[label=\alph*.]
        \item Assumption A6 is satisfied for the families of log-normal and log-student-t densities.
        \item Assumption A6 is satisfied for the families of Weibull and Gamma densities provided that the following condition is met:
        $$\frac{\pi_p(x)\tilde\lambda(x)^{k(x)}}{\pi_{\tilde p}(x)\lambda(x)^{k(x)}}=1 \quad\forall x \Longleftrightarrow \beta_{U}=\tilde\beta_{U} \text{ and } p=\tilde p,$$
        where $\lambda(x)>0$ and $k(x)>0$ are the scale and shape parameters of the distribution depending on $\beta_{U}$.
        \item Assumption A7 is satisfied for the Frank, Gumbel, Joe, Gaussian and Clayton(90) copulas, independently of the marginal distributions. For the Clayton(180) and Clayton(270) copulas we need, respectively, the following additional conditions:
        $$f_{\beta_{U}}(0|x) \neq 0 \quad\forall\beta_{U}, \forall x$$
        and
        $$f_{\beta_{C}}(0|x) < \infty \quad\forall\beta_{C}, \forall x \quad\text{and}\quad f_{\beta_{U}}(0|x) \neq 0 \quad\forall\beta_{U}, \forall x.$$
    \end{enumerate}
\end{Theorem}
\begin{Remark}\label{RmkIdent}
    The condition stated in Theorem \ref{IdentA6A7}.b is verified if there is at least one continuous covariate and the shape parameter is such that $$k(x)\neq\frac{\log\left(\frac{\pi_{\tilde p}(x)}{\pi_p(x)}\right)}{\log\left(\frac{\tilde\lambda(x)}{\lambda(x)}\right)} \quad\forall x.$$
    This requirement is purely theoretical and does not entail practical restrictions, as this form is not a natural or commonly used specification for modelling a shape parameter.
\end{Remark}
The proof of Theorem \ref{IdentA6A7} and Remark \ref{RmkIdent} are provided in the Appendix B.

Taken together, these results establish the identifiability of the model for a wide range of combinations of copulas and marginal distributions, thereby ensuring substantial flexibility and practical applicability.

%%%%%%%%%%%%%%%%%%%%%
\subsection{Parameter estimation}\label{subsect3ParamEst}
%%%%%%%%%%%%%%%%%%%%%

Given an i.i.d. sample $\mathcal{D}=\left(Y_i, \Delta_{Ti},\Delta_{Ci}, X_i\right)_{i=1,\cdots,n}$, the individual likelihood contributions derived at the beginning of Subsection \ref{subsect3Ident} yield the following log-likelihood function:
\begin{align*}
   \ell(\beta;\mathcal{D}) & = \sum_{i=1}^n \log\left(l(\beta; Y_i, \Delta_{Ti}, \Delta_{Ci}, X_i)\right)\\
   & = \sum_{i=1}^n \Delta_{Ti}\log\left(\pi_p(X_i)f_{\beta_{U}}(Y_i|X_i)\left[1-h_{\tau}^{1}\left\{F_{\beta_{C}}(Y_i|X_i)|\pi_p(X_i)F_{\beta_{U}}(Y_i|X_i),X_i\right\}\right]\right)\\
   & \quad + \sum_{i=1}^n \Delta_{Ci} \log\left(f_{\beta_{C}}(Y_i|X_i)\left[1-h_{\tau}^{2}\left\{\pi_p(X_i)F_{\beta_{U}}(Y_i|X_i)|F_{\beta_{C}}(Y_i|X_i),X_i\right\}\right]\right)\\
   & \quad + \sum_{i=1}^n \left(1-\Delta_{Ti}-\Delta_{Ci}\right)\log\left(\overline{\C}_{\tau}\left\{1-\pi_p(X_i)F_{\beta_{U}}(Y_i|X_i), 1-F_{\beta_{C}}(Y_i|X_i)\big|X_i\right\}\right).
\end{align*}

Maximising this function with respect to $\beta$ yields the maximum likelihood estimator $\hat{\beta}$, whose asymptotic properties follow from standard maximum likelihood theory. Details on the required regularity conditions, including those related to limits and first- and second-order derivatives of the log-likelihood, can be found in Chapter $8$ of \cite{Mittelhammer}.
%Standard errors and confidence intervals for each component of $\hat{\beta}$ can be obtained using the usual theory based on the Fisher information matrix.

%%%%%%%%%%%%%%%%%%%%%%%%%%%%%%%%%%%%%%%%%%%%%%%%%%%%%%%%%%%%%%%%%%%%%
\section{Simulation study} \label{sect4Simus}
%%%%%%%%%%%%%%%%%%%%%%%%%%%%%%%%%%%%%%%%%%%%%%%%%%%%%%%%%%%%%%%%%%%%%

This section examines the finite-sample performance of the proposed estimation procedure. In the first subsection, we assess its behaviour across a range of settings, varying the copula family, marginal distributions, sample size, and strength of dependence, in order to evaluate their impact on the results. The second subsection is devoted to an analysis of the robustness of the method.

%%%%%%%%%%%%%%%%%%%%%
\subsection{Performance of the method}\label{subsect4Perf}
%%%%%%%%%%%%%%%%%%%%%

The aim of this subsection is to assess the ability of the model to accurately estimate its parameters through simulation. To this end, we perform simulations by generating datasets and maximising the likelihood using the \texttt{optim}$(\cdot)$ function in \textsf{R} (\cite{R}). The BFGS algorithm is employed, with a maximum of 100 iterations.

To evaluate the influence of various factors on the estimation results, we define a reference scenario and then modify its characteristics one at a time. In the reference setting, both $U$ and $C$ follow non-truncated log-normal distributions. The dependence between $T$ and $C$ does not depend on the covariate, is strong and positive ($\tau^{K}=0.75$) and is modelled using a Frank copula. No administrative censoring is introduced. The sample size is $1000$, and a single continuous covariate positively affects the model parameters, with slopes equal to $0.7$ in the linear relations, except for the dependence parameter since it is assumed to not depend on the covariate.

All scenarios considered in this subsection are summarised in Table \ref{tblSettings}, where deviations from the reference scenario (1) are highlighted in bold. The first two columns specify the marginal distributions of $U$ and $C$. For instance, in scenario (6), both marginals are Weibull distributions with the upper $1\%$ of the support of $U$ truncated. This truncation point can be chosen arbitrarily within the support of $U|X$ and may therefore be set sufficiently large so that the proportion of truncated observations is negligible. The next two columns describe the dependence between $T$ and $C$. In scenario (10), a range of values is provided for $\tau^{K}_{\tau}(x)$ because the covariate also affects this parameter, so it is no longer constant. The fifth column indicates whether administrative censoring is present. In scenario (8), $A$ is generated from a $Weibull(shape=20, scale=2.5)$ distribution. The sixth column reports the sample size, which is $1000$ in all cases except scenario (4), where it is $500$. The final two columns relate to the covariate. It is continuous in all scenarios except (9), where it is binary. Its effect on the model parameters is generally substantial, except in scenario (7), where the slopes in the linear relations are reduced to $0.4$.

Each scenario considered satisfies the identifiability conditions of the model. Specifically, assumptions A1--A3 are verified by Theorem 2 of \cite{DelhelleVK}. The non-truncated log-normal distribution satisfies assumption A6 by Theorem \ref{IdentA6A7}(a), while the truncated Weibull distribution satisfies assumption A4, as established in Theorem 2 of \cite{DelhelleVK}. Regarding the copulas, the Frank and Joe families satisfy, independently of the marginal distributions, either assumption A5 by Theorem 2 of \cite{DelhelleVK} or assumption A7 by Theorem \ref{IdentA6A7}(c).

\begin{table}[!h]
\centering
%\fontsize{9}{9}\selectfont
%\scriptsize
%\small
\begin{tabular}{l||cc|cc|c|c|cc}
%\hline
Scenario & Densities & Trunc. $U$ & Dep. strength & Copula & $A$ & Sample size & $X$ & Slopes \\ 
 \hline
 1 & Log-normal & No & 0.75 & Frank & No & 1000 & Cont. & 0.7 \\ 
 2 & Log-normal & No & \bf{0.25} & Frank & No & 1000 & Cont. & 0.7 \\ 
 3 & Log-normal & No & \bf{-0.75} & Frank & No & 1000 & Cont. & 0.7 \\ 
 4 & Log-normal & No & 0.75 & Frank & No & \bf{500} & Cont. & 0.7 \\ 
 5 & Log-normal & No & 0.75 & \bf{Joe} & No & 1000 & Cont. & 0.7 \\ 
 6 & \bf{Weibull} & \bf{Yes} & 0.75 & Frank & No & 1000 & Cont. & 0.7 \\ 
 7 & Log-normal & No & 0.75 & Frank & No & 1000 & Cont. & \bf{0.4} \\ 
 8 & Log-normal & No & 0.75 & Frank & \bf{Yes} & 1000 & Cont. & 0.7 \\ 
 9 & Log-normal & No & 0.75 & Frank & No & 1000 & \bf{Bin.} & 0.7 \\ 
 10 & Log-normal & No & \bf{[0.55, 0.87]} & Frank & No & 1000 & Cont. & 0.7 \\ 
\end{tabular}
\smallskip
\caption{Settings of the cases considered in this subsection. \textit{Dep. strength} is the Kendall's tau (depending on the covariate for scenario (10)). The column \textit{X} indicates if the covariate is continuous (\textit{Cont.}) or binary (\textit{Bin.}).}
\label{tblSettings}
\end{table}

Table \ref{tblPctgsCens} below reports the cure rates and the dependent and administrative censoring rates in the data for each scenario listed in Table \ref{tblSettings}.

\begin{table}[!h]
\centering
%\fontsize{9}{9}\selectfont
%\scriptsize
%\small
\begin{tabular}{l|rrrrrrrrrr}
%\hline
& 1 & 2 & 3 & 5 & 6 & 7 & 8 & 9 & 10 \\ 
 \hline
 Dependent censoring rate ($C$) & 51.1 & 52.3 & 50.9 & 52.5 & 57.9 & 50.3 & 33.7 & 49 & 70.7 \\ 
 Administrative censoring rate ($A$) & 0 & 0 & 0 & 0 & 0 & 0 & 18.8 & 0 & 0 \\ 
 Cure rate & 40.4 & 40.4 & 40.4 & 40.4 & 40.4 & 40.1 & 40.4 & 32.5 & 40.5 \\
\end{tabular}
\smallskip
\caption{Censoring and cure rates for the settings considered in this subsection.}
\label{tblPctgsCens}
\end{table}

Regarding data generation, each coefficient of the model defined in (\ref{eq:modelmixture})--(\ref{eq:modelcopula}) is fixed to values depending on the scenario under consideration; these values are provided in Table \ref{Coefs41} in Appendix A. The covariate is generated from a $Unif(-1, 1)$ distribution when continuous, and from a $Ber(0.5)$ distribution when binary. The choice of link functions is guided by the parameter domains. For Kendall's $\tau$, we use the link $\{\exp(\cdot)-1\}/\{\exp(\cdot)+1\}$ when the dependence may be positive or negative (Frank copula), and the logistic function $\{\exp{\left(\cdot\right)}\}/\{1+\exp{\left(\cdot\right)}\}$ when the dependence is constrained to be positive (Joe copula). For strictly positive parameters (e.g., Weibull parameters and the standard deviation of the logarithm in the log-normal distribution), the exponential link is used, whereas the identity link is used for the remaining log-normal parameter, whose domain is $\mathbb{R}$. And, as throughout this document, a logistic model is used for the incidence.

Starting values for the optimisation are selected from a grid of $20$ parameter vectors, combining different dependence values with random perturbations of baseline coefficient estimates. Specifically, for the dependence parameter, we generate a sequence of $20$ values between $-1$ and $1$ (or between $0$ and $1$ for the Joe copula), to which the inverse link function is applied. For the remaining parameters, initial values are obtained as follows. The incidence is first estimated using the \texttt{survfit}$(\cdot)$ function in \textsf{R}; the logit transform of this estimate is then used as an initial value for the corresponding intercept. For the distribution parameters, maximum likelihood estimates are computed without taking into account neither the cure fraction nor the dependent censoring and are then transformed via the inverse link functions to obtain initial intercepts. All slope coefficients are initially set to $1$. The intercepts are then perturbed by multiplying them by $20$ random draws from a $Unif(0.5, 1.5)$ distribution, while the slopes are perturbed using $20$ random draws from a $Unif(-1, 1)$ distribution. These components are combined to construct a grid of 20 starting values.

The simulation study is based on $1000$ replications. Owing to the complexity of the optimisation problem and the number of parameters involved, the algorithm fails to converge for a small proportion of samples. This results in a limited number of outlying estimates, primarily affecting the intercept associated with the incidence component. These outliers are removed entirely or partially, with a maximum exclusion rate of 2\% of the iterations.
%As is frequently observed in the literature, the decision was taken to remove these outliers prior to performing bias calculations and other final computations.

Selected results are reported in Tables \ref{tblSim}--\ref{tblMedX}. To keep the document concise, only a subset of the scenarios listed in Table \ref{tblSettings} is included in Table \ref{tblMedX}. For clarity, we restrict attention to the coefficients associated with the dependence, the incidence, and the parameters of the distribution of $U$. For each scenario, the true coefficient values, the corresponding bias, the standard deviation of the estimators, and the root mean squared error are reported. Table \ref{tblSim} presents the results for the coefficient estimates. Table \ref{tblMedX} reports the estimated dependence parameter ($\tau^{K}_{\tau}$), the model incidence ($\pi_p$), and the quartiles of the distribution of $U$ (denoted by Q1$_U$, Q2$_U$, and Q3$_U$), evaluated for the continuous covariate set to $0.5$.

\begin{table}[!h]
\centering
%\fontsize{9}{9}\selectfont
%\scriptsize
%\small
\begin{tabular}{c|l|rrrrrrrr}
%\hline
Scenario &  & $\beta_{U_{0}}$ & $\beta_{U_{1}}$ & $\beta_{U_{2}}$ & $\beta_{U_{3}}$ & $p_0$ & $p_1$ & $\tau_0$ & $\tau_1$ \\ 
  \hline
1 & True values & 0.00 & 0.70 & -0.69 & 0.70 & 0.41 & 0.70 & 1.95 & \\ 
  & Bias & -0.00 & -0.00 & -0.01 & -0.00 & -0.00 & -0.00 & 0.00 & \\  
  & SD & 0.04 & 0.06 & 0.05 & 0.08 & 0.12 & 0.20 & 0.24 & \\ 
  & RMSE & 0.04 & 0.06 & 0.05 & 0.08 & 0.12 & 0.20 & 0.24 & \\ 
  \hline
2 & True values & 0.00 & 0.70 & -0.69 & 0.70 & 0.41 & 0.70 & 0.51 & \\ 
  & Bias & -0.01 & -0.01 & -0.01 & -0.00 & -0.01 & -0.01 & -0.12 & \\  
  & SD & 0.06 & 0.07 & 0.06 & 0.08 & 0.18 & 0.25 & 0.47 & \\ 
  & RMSE & 0.06 & 0.07 & 0.06 & 0.08 & 0.18 & 0.25 & 0.48 & \\ 
   \hline
3 & True values & 0.00 & 0.70 & -0.69 & 0.70 & 0.41 & 0.70 & -1.95 & \\ 
  & Bias & 0.01 & -0.01 & -0.00 & -0.01 & 0.12 & -0.02 & -0.08 & \\ 
  & SD & 0.08 & 0.07 & 0.07 & 0.09 & 0.92 & 0.53 & 0.63 & \\ 
  & RMSE & 0.08 & 0.07 & 0.07 & 0.09 & 0.93 & 0.53 & 0.63 & \\ 
   \hline
4 & True values & 0.00 & 0.70 & -0.69 & 0.70 & 0.41 & 0.70 & 1.95 & \\ 
  & Bias & 0.00 & -0.00 & -0.01 & -0.01 & 0.01 & 0.01 & 0.01 & \\ 
  & SD & 0.06 & 0.09 & 0.07 & 0.11 & 0.18 & 0.30 & 0.26 & \\ 
  & RMSE & 0.06 & 0.09 & 0.07 & 0.11 & 0.18 & 0.30 & 0.26 & \\ 
   \hline
5 & True values & 0.00 & 0.70 & -0.69 & 0.70 & 0.41 & 0.70 & 1.10 & \\ 
  & Bias & -0.00 & -0.00 & -0.01 & 0.00 & 0.00 & 0.00 & -0.00 & \\ 
  & SD & 0.04 & 0.07 & 0.05 & 0.08 & 0.13 & 0.22 & 0.14 & \\ 
  & RMSE & 0.04 & 0.07 & 0.05 & 0.08 & 0.13 & 0.22 & 0.14 & \\ 
   \hline
6 & True values & 0.50 & 0.70 & 0.20 & 0.70 & 0.41 & 0.70 & 1.95 & \\ 
  & Bias & 0.01 & 0.01 & -0.00 & -0.01 & 0.02 & -0.00 & -0.00 & \\ 
  & SD & 0.04 & 0.07 & 0.07 & 0.09 & 0.13 & 0.20 & 0.17 & \\ 
  & RMSE & 0.04 & 0.07 & 0.07 & 0.09 & 0.13 & 0.20 & 0.17 & \\
   \hline
7 & True values & 0.00 & 0.40 & -0.69 & 0.40 & 0.41 & 0.40 & 1.95 & \\ 
  & Bias & -0.00 & -0.01 & -0.01 & -0.00 & 0.00 & -0.00 & -0.00 & \\ 
  & SD & 0.04 & 0.07 & 0.05 & 0.09 & 0.11 & 0.19 & 0.15 & \\ 
  & RMSE & 0.04 & 0.07 & 0.05 & 0.09 & 0.11 & 0.19 & 0.15 & \\ 
    \hline
8 & True values & 0.00 & 0.70 & -0.69 & 0.70 & 0.41 & 0.70 & 1.95 & \\ 
  & Bias & -0.00 & -0.00 & -0.01 & -0.00 & -0.00 & -0.00 & -0.02 & \\ 
  & SD & 0.05 & 0.07 & 0.05 & 0.08 & 0.14 & 0.22 & 0.37 & \\ 
  & RMSE & 0.05 & 0.07 & 0.05 & 0.08 & 0.14 & 0.22 & 0.37 & \\ 
    \hline
9 & True values & 0.00 & 0.70 & -0.69 & 0.70 & 0.41 & 0.70 & 1.95 & \\ 
  & Bias & -0.00 & -0.01 & -0.01 & -0.00 & -0.00 & 0.05 & -0.03 & \\ 
  & SD & 0.06 & 0.15 & 0.08 & 0.11 & 0.16 & 0.63 & 0.44 & \\ 
  & RMSE & 0.06 & 0.15 & 0.08 & 0.11 & 0.16 & 0.63 & 0.44 & \\ 
    \hline
10 & True values & 0.00 & 0.70 & -0.69 & 0.70 & 0.41 & 0.70 & 1.95 & 0.70 \\ 
  & Bias & 0.06 & 0.04 & 0.01 & -0.02 & -0.78 & 0.36 & -1.00 & -0.18 \\ 
  & SD & 0.07 & 0.11 & 0.07 & 0.11 & 0.16 & 0.25 & 0.25 & 0.34 \\ 
  & RMSE & 0.09 & 0.12 & 0.07 & 0.12 & 0.80 & 0.44 & 1.03 & 0.39 \\ 
\end{tabular}
\smallskip
\caption{True values, bias, standard deviation (SD) and root mean squared error (RMSE) of the estimators for the settings described above. The final column, corresponding to slope for the dependence, concerns only the scenario (10), as this is the only one in which the covariate affects $\tau^{K}_{\tau}$.}
\label{tblSim}
\end{table}

Overall, the results in Table \ref{tblSim} are satisfactory, even excellent. Biases are generally small, with the exception of the dependence parameter in scenario (2). Only scenario (10), in which the dependence parameter is also affected by the covariate, yields less favourable results. In this case, biases increase for several coefficients, particularly those associated with the incidence and $\tau^{K}_{\tau}$. One possible explanation is that, as shown in Table \ref{tblPctgsCens}, scenario (10) involves a substantially higher proportion of censored observations. The reduced number of uncensored observations limits the available information for estimating the dependence structure.

\begin{table}[!h]
\centering
%\fontsize{9}{9}\selectfont
%\scriptsize
%\small
\begin{tabular}{c|l|rrrrrrr}
%\hline
Scenario &  & $\tau^{K}_{\tau}(0.5)$ & $\pi_p(0.5)$ & $\theta_{U1}(0.5)$ & $\theta_{U2}(0.5)$ & Q1$_U(0.5)$ & Q2$_U(0.5)$ & Q3$_U(0.5)$ \\ 
  \hline
1 & True values & 0.75 & 0.68 & 0.35 & 0.71 & 0.88 & 1.42 & 2.29 \\ 
  & Bias & -0.00 & -0.00 & -0.00 & -0.01 & 0.00 & -0.00 & -0.01 \\ 
  & SD & 0.04 & 0.04 & 0.07 & 0.05 & 0.05 & 0.10 & 0.22 \\ 
  & RMSE & 0.04 & 0.04 & 0.07 & 0.05 & 0.05 & 0.10 & 0.22 \\ 
  \hline
3 & True values & -0.75 & 0.68 & 0.35 & 0.71 & 0.88 & 1.42 & 2.29 \\ 
  & Bias & -0.00 & 0.00 & 0.00 & -0.00 & 0.01 & 0.01 & 0.02 \\ 
  & SD & 0.07 & 0.07 & 0.09 & 0.05 & 0.06 & 0.15 & 0.32 \\ 
  & RMSE & 0.07 & 0.07 & 0.09 & 0.05 & 0.06 & 0.15 & 0.32 \\  
  \hline
5 & True values & 0.75 & 0.68 & 0.35 & 0.71 & 0.88 & 1.42 & 2.29 \\ 
  & Bias & -0.00 & -0.00 & -0.00 & -0.00 & 0.00 & -0.00 & 0.00 \\ 
  & SD & 0.03 & 0.05 & 0.07 & 0.05 & 0.05 & 0.10 & 0.23 \\ 
  & RMSE & 0.03 & 0.05 & 0.07 & 0.05 & 0.05 & 0.10 & 0.23 \\ 
  \hline
6 & True values & 0.75 & 0.68 & 2.34 & 1.73 & 1.02 & 1.48 & 1.99 \\ 
  & Bias & -0.00 & 0.00 & 0.04 & -0.01 & 0.00 & -0.00 & -0.01 \\ 
  & SD & 0.04 & 0.02 & 0.11 & 0.06 & 0.04 & 0.05 & 0.07 \\ 
  & RMSE & 0.04 & 0.02 & 0.12 & 0.06 & 0.04 & 0.05 & 0.07 \\ 
 \hline
10 & True values & 0.82 & 0.68 & 0.35 & 0.71 & 0.88 & 1.42 & 2.29 \\ 
  & Bias & -0.28 & -0.14 & 0.08 & 0.00 & 0.07 & 0.13 & 0.22 \\ 
  & SD & 0.08 & 0.06 & 0.11 & 0.06 & 0.08 & 0.18 & 0.37 \\ 
  & RMSE & 0.29 & 0.16 & 0.14 & 0.06 & 0.11 & 0.22 & 0.43 \\
\end{tabular}
\smallskip
\caption{True values, bias, standard deviation (SD) and root mean squared error (RMSE) of the parameters when the continuous covariate is equal to $0.5$, for several settings described above. Scenario (10), is the only one in which the covariate affects $\tau^{K}_{\tau}$ as well.}
\label{tblMedX}
\end{table}

The results in Table \ref{tblMedX} are also satisfactory and consistent with those observed for the coefficient estimates. Again, scenario (10) stands out as more challenging. Although biases appear slightly reduced when considering the model parameters, and standard deviations are somewhat smaller, the estimation of $\tau^{K}_{\tau}$ and the incidence remain difficult when the dependence parameter varies with the covariate.

Taken together, these findings indicate that the proposed estimation procedure exhibits good finite-sample performance overall.

%%%%%%%%%%%%%%%%%%%%%
\newpage
\subsection{Robustness}\label{subsect4Robustness}
%%%%%%%%%%%%%%%%%%%%%

This subsection is devoted to assessing the robustness of the proposed method to the misspecification of the marginals
or the copula. To this end, we consider several scenarios. In some cases, the marginal distributions of $U$ and $C$ are fixed while the copula is selected from the data; in others, the copula is fixed and the marginal distributions are chosen using a model selection criterion.

Copula selection is performed among the independence copula, Frank, Gumbel, Joe, Clayton(90), Clayton(180), Clayton(270), and Gaussian families, while marginal distributions are selected from the Weibull, Gamma, and log-normal families. In each case, the optimisation is carried out and, among all candidate models, the best one is chosen using the Akaike Information Criterion (AIC).

These approaches are intended to reflect more realistic applications, where the underlying dependence structure and marginal distributions are unknown, as in the analysis of real data described in Section \ref{sect5Data}. They also allow us to evaluate the ability of the model selection procedure to recover the true copula and marginal distributions, as well as to assess the sensitivity of relevant quantities to model misspecification.

For these analyses, $500$ replications are used, with a sample size of $1000$. Two settings are considered. In the first, data are generated from a truncated Weibull marginal distribution for $U$ and a non-truncated Weibull for $C$, with dependence between $T$ and $C$ modelled by a Joe copula with $\tau^{K}=0.75$. In the second, data are generated from non-truncated log-normal marginal distributions, similarly to scenario (5) in Subsection \ref{subsect4Perf}. In both settings, no administrative censoring is introduced, and a uniform covariate affects all model parameters except the dependence parameter. The corresponding coefficient values are provided in Table \ref{Coefs42} in Appendix A.

The results are reported in Table \ref{tblChosenWbTrLn}. For each setting, we compute the proportion of correct selections based on the AIC, as well as the relative bias and root mean squared error for the estimated quartiles of the survival function of $U$, evaluated for the covariate equal to $0.5$.

\begin{table}[!h]
\centering
%\fontsize{9}{9}\selectfont
%\scriptsize
%\small
\begin{tabular}{l|c|l|rrr}
& \% correct &  & Q1$_U(0.5)$ & Q2$_U(0.5)$ & Q3$_U(0.5)$ \\ 
 \hline
 Correctly specified model & & Relative bias & 0.004 & -0.003 & -0.009 \\ 
 (Joe copula and truncated Weibull) & & RMSE & 0.032 & 0.037 & 0.053 \\ 
 \hline
 \hline
 \multirow{2}{*}{Data-driven densities} & \multirow{2}{*}{99.0}
 & Relative bias & 0.004 & -0.003 & -0.008 \\ 
 & & RMSE & 0.032 & 0.037 & 0.053 \\  
 \hline
 \hline
 \multirow{2}{*}{Data-driven copula} & \multirow{2}{*}{54.4}
 & Relative bias & 0.004 & -0.003 & -0.009 \\ 
 & & RMSE & 0.032 & 0.037 & 0.053 \\ 
 \hline
 \hline
 \hline
 Correctly specified model & & Relative bias & 0.001 & 0.001 & 0.002 \\ 
 (Joe copula and non-truncated log-normal) & & RMSE & 0.052 & 0.113 & 0.246 \\ 
 \hline
 \hline
 \multirow{2}{*}{Data-driven densities} & \multirow{2}{*}{83.8}
 & Relative bias & -0.004 & -0.010 & -0.019 \\ 
 & & RMSE & 0.052 & 0.116 & 0.272 \\ 
 \hline
 \hline
 \multirow{2}{*}{Data-driven copula} & \multirow{2}{*}{57.6}
 & Relative bias & -0.001 & -0.002 & -0.003 \\ 
 & & RMSE & 0.053 & 0.115 & 0.251 \\
\end{tabular}
\smallskip
\caption{Relative bias and root mean squared error (RMSE) for the estimated quartiles when $X=0.5$.}
\label{tblChosenWbTrLn}
\end{table}

The truncated Weibull case shows that the model selection procedure performs well in identifying the marginal distributions. In contrast, the copula is incorrectly selected in approximately 46\% of the cases. A more detailed analysis indicates that the AIC predominantly selects the Clayton(180) copula in these instances, which is not unexpected given its similarity to the Joe copula. Despite these misclassifications, the relative biases and root mean squared errors are nearly identical to those obtained under the correctly specified model. Similar patterns are observed for the case of non-truncated log-normal distributions. Again, confusion arises primarily between the Joe and Clayton(180) copulas. Regarding the marginal distributions, an error rate of approximately 16\% is observed, with most misclassifications occurring between the log-normal and Gamma families.

These results indicate that, although the combination of a Joe copula with non-truncated log-normal marginals constitutes a deliberately challenging setting for model selection, the estimated survival curves remain accurate. This suggests that the proposed method possesses a satisfactory degree of robustness with respect to the misspecification of the copula family and the marginal distributions.

%%%%%%%%%%%%%%%%%%%%%%%%%%%%%%%%%%%%%%%%%%%%%%%%%%%%%%%%%%%%%%%%%%%%%
\section{Illustration on real data} \label{sect5Data}
%%%%%%%%%%%%%%%%%%%%%%%%%%%%%%%%%%%%%%%%%%%%%%%%%%%%%%%%%%%%%%%%%%%%%

This section illustrates the practical relevance of the proposed approach through its application to a real dataset. The data are obtained from the \cite{OvaryP2} (www.seer.cancer.gov), and focus on patients diagnosed with ovarian cancer. From this heterogeneous dataset, a consistent subset is selected. We restrict our attention to patients diagnosed with malignant cancer in 2000 or 2001, with an age at diagnosis between $45$ and $84$ years. Among the many available variables, the following are retained:
\begin{enumerate}[label=(\arabic*)]
    \item \textit{Survival months}, corresponding to the observed variable $Y$. Individuals with missing values or with a recorded survival time equal to zero are excluded.
    \item \textit{Cause of death (COD) to site recode}, which records the cause of death. As this is not the primary focus of the analysis, the variable is recoded into three categories: alive (treated as administratively censored), death due to ovarian cancer (event of interest), and death due to other causes (potentially dependent censoring).
    \item \textit{Age} at diagnosis, taking values between $45$ and $84$ years. (Scaled between -1 and 1 for practical use in code.)
    \item \textit{Stage}, a categorical variable with three levels indicating disease progression: \textit{Localized} (cancer confined to the organ of origin, $13.83\%$), \textit{Regional} (spread to nearby structures or lymph nodes, $17.45\%$), and \textit{Distant} (metastasis to distant organs, $68.72\%$).
\end{enumerate}
The dataset retained for the subsequent analyses comprises $7518$ patients. Among them, $4636$ ($61.66\%$) died from ovarian cancer, $1258$ ($16.73\%$) were administratively censored (i.e., still alive at the end of follow-up), and the remaining $1624$ ($21.60\%$) died from causes other than ovarian cancer and are therefore subject to potentially dependent censoring. For example, after ovarian cancer itself, the most frequent cause of death recorded in the database is \textit{Diseases of Heart}, which may plausibly be indirectly associated with ovarian cancer through the pro-thrombotic state induced by the disease (see, for example, \cite{OvarianVTE2005}). This interpretation is further supported by evidence showing that venous thromboembolism is associated with reduced median overall survival among patients with ovarian cancer (\cite{OvarianVTE2016}).

Although the Kaplan–Meier estimator is not appropriate in the presence of dependent censoring, plotting it without distinguishing between censoring types reveals a plateau around $0.3$ in Figure \ref{fig:KM}. This feature is commonly interpreted as evidence of a cure fraction (\cite{AmicoVK}).

\begin{figure}[!h]
    \centering
    \vspace*{-0.5cm}
    \includegraphics[scale=0.6]{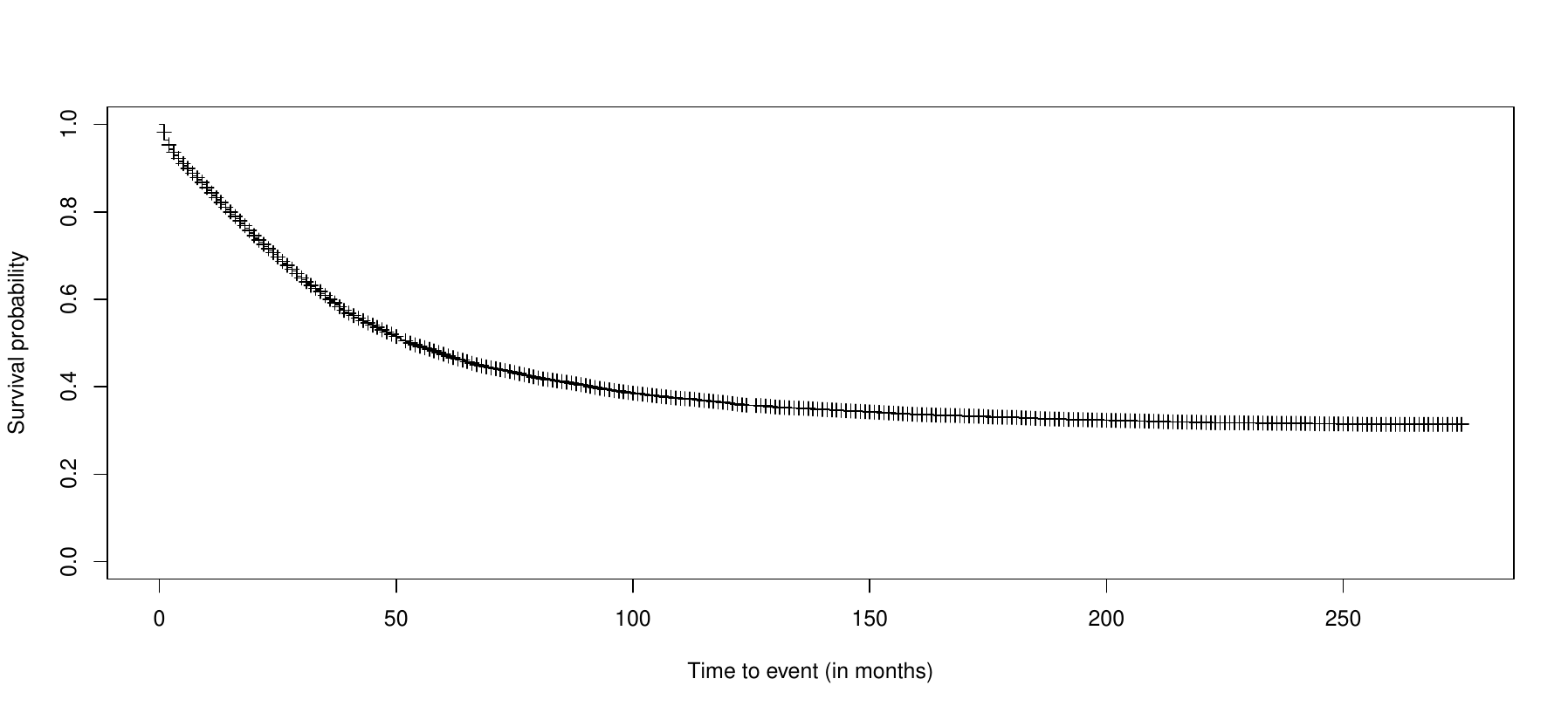}
    \caption{Kaplan-Meier curve for the data, with the two types of censoring being lumped together.}
    \label{fig:KM}
\end{figure}

The coexistence of a cure fraction and potentially dependent censoring suggests that the proposed model is well suited for the analysis of these data.

The analysis was carried out using the proposed model under various combinations of copulas and marginal distributions. The copula families considered include the independence copula, as well as the Frank, Gumbel, Joe, Clayton(90-180-270), and Gaussian copulas. For the marginal distributions of $U$ and $C$, for simplicity, the same family is used for both variables in each specification, although this is not required in practice. The distributions considered are non-truncated and include the Weibull, Gamma, and log-normal families. The link functions used for parameter modelling are identical to those employed in the simulation study; a complete specification is provided in Appendix A.

For each candidate model, the likelihood and the AIC are computed. The model with the lowest AIC ($70665.52$) corresponds to a Frank copula with Weibull marginals. This selected model is shown to be significantly preferable to the model assuming independence between $T$ and $C$, based on a likelihood ratio test. The test statistic is
$$\lambda_{LR} = -2(\text{logLik}_{Ind}-\text{logLik}_{Frank}) = -2(-35377.14+35311.76) = 130.76,$$
which substantially exceeds the critical value $\chi_{1,0.95}^2=3.84$. The associated p-value is
$$P(\chi_{1}^2>\lambda_{LR})=P(\chi_{1}^2>130.76)=2.8e^{-30},$$
providing strong evidence in favour of the model with the Frank copula. This result supports the presence of dependence between the survival time and the non-administrative censoring time.

According to the AIC criterion, the next two best models are the Gumbel copula with Weibull marginal (AIC=$70677.94$) and the Frank copula with Gamma marginals (AIC=$70699.44$). Figure \ref{fig:SurvCurvTop3}, which displays the survival functions for uncured individuals across distant, regional, and localized cancer stages, shows that these models produce results very close to those of the selected model. This similarity is expected, as they share either the copula family or the marginal distributions. Such consistency further supports the robustness of the proposed approach.

\begin{figure}[!h]
    \centering
    \vspace*{0.25cm}
    \includegraphics[scale=0.8]{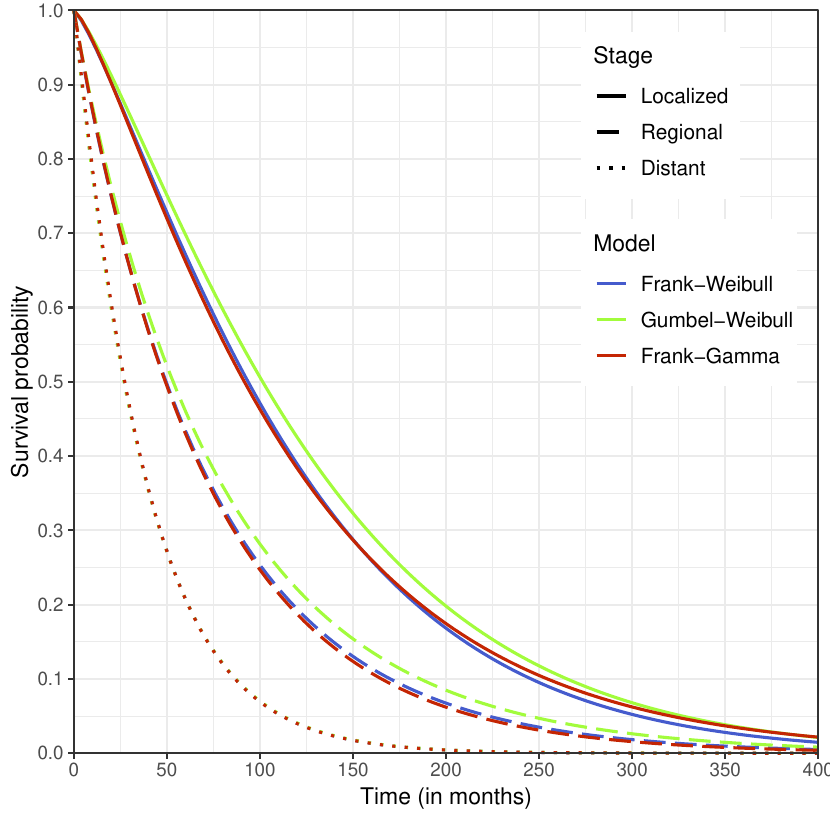}
    \caption{Comparison of the survival functions for uncured individuals across the three best models according to the AIC criterion.}
    \label{fig:SurvCurvTop3}
\end{figure}

The contribution of each covariate to the different model parameters was assessed by fitting a series of reduced models in which the effect of a given covariate on a specific parameter was removed (e.g., excluding the effect of \textit{Stage} on the incidence, then excluding that of \textit{Age}, and so on). The AIC values of these reduced models were compared with that of the full model. The only reduced model yielding a lower AIC was the one in which \textit{Age} had no effect on the shape parameter of the distribution of $C$.% However, a likelihood ratio test indicated no significant difference between this reduced specification and the full model. Consequently, the full model is retained for the subsequent analyses.

Table \ref{tblOvary} reports the estimated parameters for the selected model, evaluated at the median value of the \textit{Age} variable (63 years) and for each of the three cancer stages included in the dataset. The quartiles of the distribution of $U$ are also provided.

\begin{table}[!h]
\centering
%\fontsize{8}{8}\selectfont
\small
%\scriptsize
\begin{tabular}{l|rrrrrrrrr}
 Stage & $\tau^{K}$ & $p$ & $\theta_{U1}$ & $\theta_{U2}$ & $\theta_{C1}$ & $\theta_{C2}$ & Q1$_U$ & Q2$_U$ & Q3$_U$ \\
 \hline
 Localized & 0.45 & 0.27 & 1.24 & 125.70 & 1.55 & 304.27 & 46.19 & 93.63 & 163.44 \\ 
 Regional & 0.45 & 0.54 & 0.97 & 72.20 & 1.30 & 260.11 & 20.08 & 49.55 & 100.98 \\ 
 Distant & 0.45 & 0.89 & 1.03& 38.74 & 1.00 & 161.38 & 11.60 & 27.17 & 53.14 \\
\end{tabular}
\smallskip
\caption{Estimators of the parameters of the selected model (Frank - Weibull) in case of individuals aged 63 (median of \textit{Age}) and the three levels of \textit{Stage}.}
\label{tblOvary}
\end{table}

The results indicate a positive dependence between $T$ and $C$ ($\tau^{K}=0.45$), which is consistent with expectations. In addition, a cure fraction is observed in all three stages, with a marked decrease as the disease progresses to more advanced stages, in line with clinical knowledge. Similarly, the quartiles of the distribution of $U$ decrease from localized to regional and further to distant stages. This pattern is illustrated in the left panel of Figure \ref{fig:SurvCurvIncidFkWb}, which displays the survival curves for uncured individuals aged 63. The survival probability declines substantially for patients with distant-stage cancer, whereas individuals with localized-stage cancer exhibit more favourable outcomes.

A more detailed assessment of the covariate effects can also be conducted. For example, the right panel of Figure \ref{fig:SurvCurvIncidFkWb} presents the age-specific incidence rates across the different cancer stages.

\begin{figure}[!h]
    \centering
    \vspace*{-0.5cm}
    \includegraphics[scale=0.77]{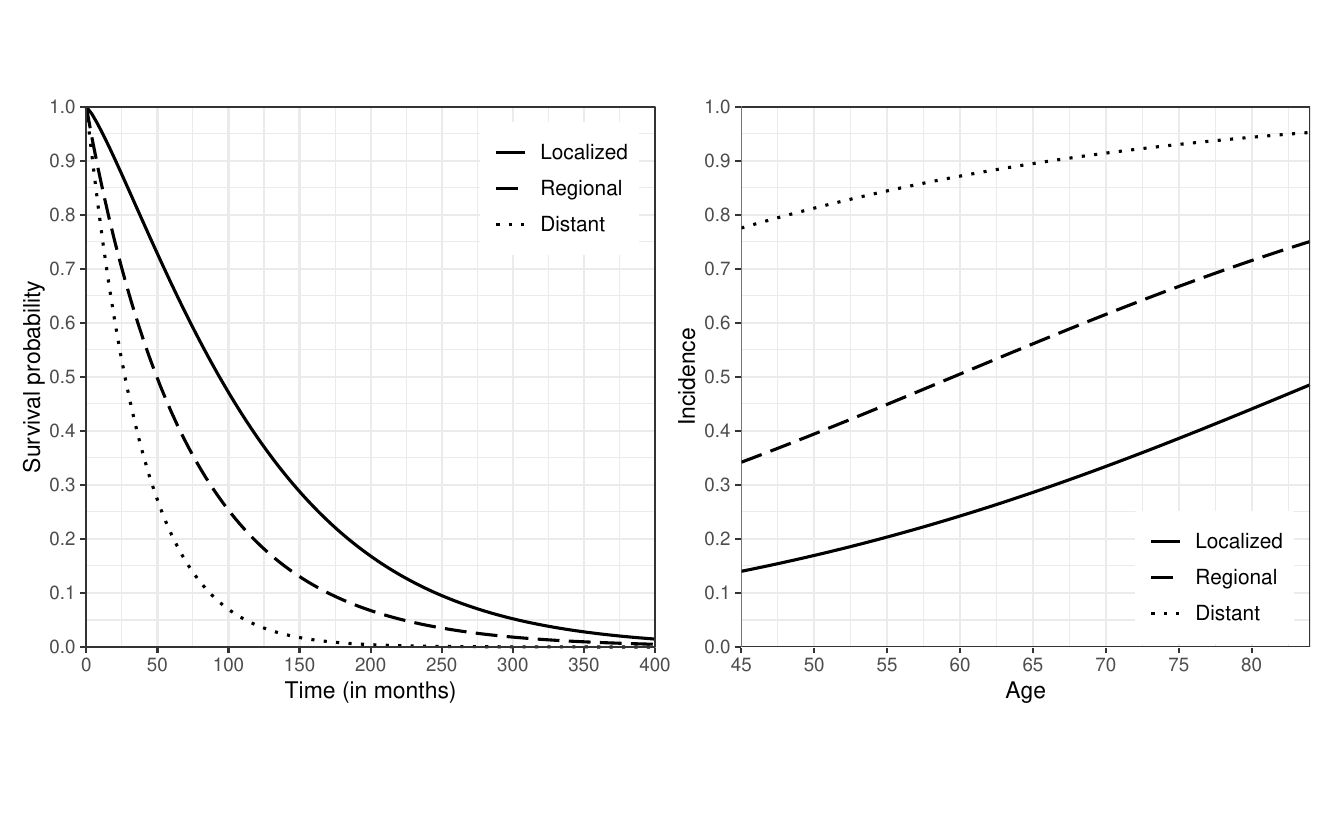}
    \caption{Left : Survival curves for uncured individuals of 63 years old (median age). Right : Age-specific incidence rates for the different stages of cancer.}
    \label{fig:SurvCurvIncidFkWb}
\end{figure}

Finally, a bootstrap analysis was conducted to illustrate an advantage of the proposed model, namely its ability to provide more accurate estimates of clinically relevant quantities than a model assuming independence. Specifically, $1000$ bootstrap samples were generated by sampling with replacement from the original dataset. For each sample, the quartiles of the distribution of $U$ were estimated for patients aged $63$, across all three cancer stages, using both a model with an independence copula and Weibull marginals, and the model selected by the algorithm (Frank copula with Weibull marginals). %Based on these results, the estimated standard deviations of the differences in quartile estimates, denoted by $\hat{\sigma}_{\text{Diff}}$, were computed, and $95\%$ confidence intervals of the form $\text{IC}_{95\%}=\text{Diff}_{Q}\pm z_{0.975}\hat{\sigma}_{\text{Diff}}$ were constructed. 
Based on these estimates, $95\%$ confidence intervals for each difference between quartiles were constructed using the percentile bootstrap method. For patients with distant-stage cancer, these intervals, which can be found in Table \ref{tblICsQunatiles}, do not include zero, indicating that the quartile estimates obtained under the Frank copula model differ significantly (at the $5\%$ significance level) from those obtained under the independence model. For instance, the median survival time is estimated at $27.17$ months with the proposed model, which is $1.25$ months shorter than the estimate obtained under the independence assumption.

%\begin{table}[!h]
%\centering
%\fontsize{10}{10}\selectfont
%\small
%\scriptsize
%\begin{tabular}{l|ccc}
% Stage & Q1$_U$ & Q2$_U$ & Q3$_U$ \\
% \hline
% Localized & [-159.95, 137.93] & [-532.53, 489.37] & [-1444.03, 1370.73] %\\ 
% Regional & [-7.45, 3.67] & [-26.56, 18.32] & [-73.58, 58.57] \\ 
% Distant & [0.23, 0.87] & [0.69, 1.82] & [1.28, 3.51] \\
%\end{tabular}
%\smallskip
%\caption{\textcolor{red}{Décider si on l'intègre au document final.}}
%\label{tblICsBootstrap}
%\end{table}

\begin{table}[!h]
\centering
%\fontsize{10}{10}\selectfont
\small
%\scriptsize
\begin{tabular}{l|ccc}
 Stage & Q1$_U$ & Q2$_U$ & Q3$_U$ \\
 \hline
 %Localized & [-219.86, 11.01] & [-601.39, 30.86] & [-1362.40, 72.58] \\ 
 %Regional & [-2.76, 1.98] & [-6.39, 5.39] & [-12.32, 12.21] \\ 
 Distant & [0.40, 1.13] & [0.96, 2.33] & [1.94, 4.06] \\
\end{tabular}
\smallskip
\caption{$95\%$ CI for the differences between quartiles of the survival time for patient with a distant-stage cancer.}
\label{tblICsQunatiles}
\end{table}

%%%%%%%%%%%%%%%%%%%%%%%%%%%%%%%%%%%%%%%%%%%%%%%%%%%%%%%%%%%%%%%%%%%%%
\section{Discussion and future research} \label{sect6Conclu}
%%%%%%%%%%%%%%%%%%%%%%%%%%%%%%%%%%%%%%%%%%%%%%%%%%%%%%%%%%%%%%%%%%%%%

In conclusion, the model proposed in this paper constitutes a substantial advancement in the literature on cure models in presence of dependent censoring. By explicitly modelling the dependence between survival and censoring times more accurately and incorporating covariate information, the proposed approach enables more precise and individualised estimation of the cure fraction, the dependence structure, and other clinically relevant quantities.

Several avenues for future research may be considered. In particular, the incorporation of variable-selection procedures, via, e.g., a penalised likelihood, could facilitate the identification of the most relevant covariates and further improve the practical applicability of the model.

%%%%%%%%%%%%%%%%%%%%%%%%%%%%%%%%%%
\section*{Declarations}
%%%%%%%%%%%%%%%%%%%%%%%%%%%%%%%%%%

%%%
\subsection*{Funding}
%%%

Morine Delhelle, Anouar El Ghouch and Ingrid Van Keilegom acknowledge the support of the ARC project (Projet d'Actions de Recherche Concertées) `Imperfect data : From mathematical foundations to applications in life sciences' of the `Communauté fran\c{c}aise de Belgique', granted by the `Académie universitaire Louvain' (2020-2025).

%%%
\subsection*{Conflict of interest}
%%%

The authors declare that they have no conflict of interest.

%%%%%%%%%%%%%%%%%%%%%%%%%%%%%%%%%%
\section*{Acknowledgment}
%%%%%%%%%%%%%%%%%%%%%%%%%%%%%%%%%%

Computational resources have been provided by the supercomputing facilities of the Université catholique de Louvain (CISM/UCL) and the Consortium des Équipements de Calcul Intensif en Fédération Wallonie Bruxelles (CÉCI) funded by the Fond de la Recherche Scientifique de Belgique (F.R.S.-FNRS) under convention 2.5020.11 and by the Walloon Region. ChatGPT has been used for language improvements.

%%%%%%%%%%%%%%%%%%%%%%%%%%%%%%%%%%
%\section*{Supplementary Material}
%%%%%%%%%%%%%%%%%%%%%%%%%%%%%%%%%%

%%%%%%%%%%%%%%%%%%%%%%%%%%
%\section*{References}
%%%%%%%%%%%%%%%%%%%%%%%%%%

\bibliographystyle{dcu}

\bibliography{References}

%%%%%%%%%%%%%%%%%%%%%%%%%%%%%%%%%%
\newpage
\section*{Appendix A : Further details} \label{appCoefs}
%%%%%%%%%%%%%%%%%%%%%%%%%%%%%%%%%%

This appendix contains all the coefficients and link functions used for the simulations and data analysis sections.

\subsection*{Coefficients for Subsection \ref{subsect4Perf}}

Here are the values of all the coefficients used to generate the data in each scenario in Subsection \ref{subsect4Perf}.

\begin{table}[!h]
\centering
%\fontsize{7}{7}\selectfont
\small
%\scriptsize
\begin{tabular}{r|rrrrrrrrrrrr}
  & $\tau_0$ & $\tau_1$ & $p_0$ & $p_1$ & $\beta_{U_{0}}$ & $\beta_{U_{1}}$ & $\beta_{U_{2}}$ & $\beta_{U_{3}}$ & $\beta_{C_{0}}$ & $\beta_{C_{1}}$ & $\beta_{C_{2}}$ & $\beta_{C_{3}}$ \\
  \hline
 1 - 4 - 8 - 9 & 1.95 & & 0.41 & 0.7 & 0 & 0.7 & -0.69 & 0.7 & 0.5 & 0.7 & -1.39 & 0.7 \\
 2 & 0.51 & & 0.41 & 0.7 & 0 & 0.7 & -0.69 & 0.7 & 0.5 & 0.7 & -1.39 & 0.7 \\
 3 & -1.95 & & 0.41 & 0.7 & 0 & 0.7 & -0.69 & 0.7 & 0.5 & 0.7 & -1.39 & 0.7 \\
 5 & 1.10 & & 0.41 & 0.7 & 0 & 0.7 & -0.69 & 0.7 & 0.5 & 0.7 & -1.39 & 0.7 \\
 6 & 1.95 & & 0.41 & 0.7 & 0.5 & 0.7 & 0.2 & 0.7 & 0.7 & 0.7 & 0.7 & 0.7 \\
 7 & 1.95 & & 0.41 & 0.4 & 0 & 0.4 & -0.69 & 0.4 & 0.5 & 0.4 & -1.39 & 0.4 \\
  10 & 1.95 & 0.7 & 0.41 & 0.7 & 0 & 0.7 & -0.69 & 0.7 & 0.5 & 0.7 & -1.39 & 0.7 \\
\end{tabular}
\smallskip
\caption{Coefficients fixed for the simulation study.}
\label{Coefs41}
\end{table}

\subsection*{Coefficients for Subsection \ref{subsect4Robustness}}

Table \ref{Coefs42} shows the values of the coefficients used to generate the data for the robustness assessment of the method.

\begin{table}[!h]
\centering
%\fontsize{7}{7}\selectfont
\small
%\scriptsize
\begin{tabular}{r|rrrrrrrrrrr}
  & $\tau_0$ & $p_0$ & $p_1$ & $\beta_{U_{0}}$ & $\beta_{U_{1}}$ & $\beta_{U_{2}}$ & $\beta_{U_{3}}$ & $\beta_{C_{0}}$ & $\beta_{C_{1}}$ & $\beta_{C_{2}}$ & $\beta_{C_{3}}$ \\
  \hline
 Joe - Truncated Weibull & 1.1 & 0.41 & 0.7 & 0.5 & 0.7 & 0.2 & 0.4 & 0.7 & 0.5 & 0.7 & 0.5 \\
 Joe - Non-truncated log-normal & 1.1 & 0.41 & 0.7 & 0 & 0.7 & -0.69 & 0.7 & 0.5 & 0.7 & -1.39 & 0.7 \\
\end{tabular}
\smallskip
\caption{Coefficients fixed for the study of the robustness.}
\label{Coefs42}
\end{table}

\subsection*{Link functions for sections \ref{sect4Simus} and \ref{sect5Data}}

The list below shows the link functions defined for each model parameter used in sections \ref{sect4Simus} and \ref{sect5Data}.

\begin{itemize}
    \item Kendall's tau for Frank and Gaussian copulas : $g(y)= \frac{\exp(y)-1}{\exp(y)+1}$
    \item Incidence and Kendall's tau for Gumbel, Joe and Clayton(180) copulas : $g(y) = \frac{\exp{\left(y\right)}}{1+\exp{\left(y\right)}}$
    \item Kendall's tau for Clayton(90) and Clayton(270) copulas : $g(y) = -\frac{\exp{\left(y\right)}}{1+\exp{\left(y\right)}}$
    \item Parameters for Weibull and Gamma distributions, and standard deviation of the logarithm in the log-normal distribution: $g(y)=\exp(y)$
    \item Mean of the logarithm in the log-normal distribution: $g(y)=y$
\end{itemize}

%%%%%%%%%%%%%%%%%%%%%%%%%%%%%%%%%%
\newpage
\section*{Appendix B : Proofs} \label{appProofs}
%%%%%%%%%%%%%%%%%%%%%%%%%%%%%%%%%%

This appendix contains the proofs for the results presented in this document.

We will start with the developments that led to the individual likelihood contributions in Subsection \ref{subsect3Ident}.

\begin{proof}[Proof of individual likelihood contributions]

First, let's compute the cumulative distribution functions:
\begin{align*}
    \Proba(Y \le y, \Delta_T=1, \Delta_C=0| X=x) &= \Proba(T \le y, T \le C, T \le A| X=x) \\
    &= \int_{0}^{y} \Proba(C \ge t, A \ge t | X=x, T=t)f_T(t|x) dt \\
    &= \int_{0}^{y} \Proba(C \ge t | X=x, T=t)\Proba(A \ge t)f_T(t|x) dt \\
    &= \int_{0}^{y} \{1-F_{C|T,X}(t|t,x)\}\{1-F_A(t)\}f_T(t|x) dt \\
    &= \int_{0}^{y} [1-h^{1}\{F_C(t|x)
|F_T(t|x),x\}]\{1-F_A(t)\}f_T(t|x) dt.
\end{align*}
Similarly,
\begin{align*}
    \Proba(Y \le y, \Delta_T=0, \Delta_C=1| X=x) &= \Proba(C \le y, C \le T, C \le A| X=x) \\
    &= \int_{0}^{y} [1-h^{2}\{F_T(c|x)|F_C(c|x),x\}]\{1-F_A(c)\}f_C(c|x) dc,
\end{align*}
and 
\begin{align*}
    \Proba(Y \le y, \Delta_T=0, \Delta_C=0| X=x) &= \Proba(A \le y, A \le T, A \le C| X=x) \\
    &= \int_{0}^{y} \Proba(T \ge a, C \ge a| X=x, A=a)f_{A}(a) da \\
    &= \int_{0}^{y} \Proba(T \ge a, C \ge a| X=x)f_{A}(a) da \\
    &= \int_{0}^{y} \overline{\C}\left(1-F_T(a|x), 1-F_C(a|x)\big|x\right)f_{A}(a) da.
\end{align*}

Then we take the derivatives of these expressions to obtain
\begin{align*}
    \frac{\partial}{\partial y} \Proba(Y \le y, \Delta_T=1, \Delta_C=0| X=x) &= \pi(x)f_U(y|x) \left[1-h^{1}\left\{F_C(y|x)\big|\pi(x)F_U(y|x),x\right\}\right]\{1-F_A(y)\},
\end{align*}
\begin{align*}
    \frac{\partial}{\partial y} \Proba(Y \le y, \Delta_T=0, \Delta_C=1| X=x) &= f_C(y|x) \left[1-h^{2}\left\{\pi(x)F_U(y|x)\big|F_C(y|x),x\right\}\right]\{1-F_A(y)\},
\end{align*}
and
\begin{align*}
    \frac{\partial}{\partial y} \Proba(Y \le y, \Delta_T=0, \Delta_C=0| X=x) = \overline{\C}\left(1-\pi(x)F_U(y|x), 1-F_C(y|x)\big|x\right)f_A(y).
\end{align*}
\end{proof}

We then introduce a lemma that will be used below to derive the equality of coefficients from the equality of model parameters. The proof is not provided as it is straightforward.

\begin{Lemma}\label{IdentAlpha}
    Let $g$ be a one-to-one function and $X$ a vector such that the components of $X$ are linearly independent. Then $g\left(\alpha^\top X\right)=g\left(\tilde\alpha^\top X\right)$ almost surely implies $\alpha=\tilde\alpha$.
\end{Lemma}

Next, here are the proofs of Theorem \ref{IdentMod} and Theorem \ref{IdentA6A7}, which concern the identifiability of the model.

\begin{proof}[Proof of Theorem \ref{IdentMod}]
For the combination of assumptions A1 to A5, the proof follows from Theorem 1 in \cite{DelhelleVK}, which states that if these assumptions are verified, then the parameters of the model are identified for each fixed $x$. Since the link functions are one-to-one and the components of $X$ are linearly independent, all coefficients in $\beta$ are directly identified by applying Lemma \ref{IdentAlpha}%\\
%Let us analyse the proof for the incidence parameter in more detail. Theorem 1 in \cite{DelhelleVK} gives that $\pi_p(X) = \pi_{\tilde p}(X)$, or equivalently $\exp{\left(p^\top X\right)}/(1+\exp{\left(p^\top X\right)}) = \exp{\left(\tilde p^\top X\right)}/(1+\exp{\left(\tilde p^\top X\right)})$. Since the logistic function is one-to-one, applying Lemma \ref{IdentAlpha} directly leads to $\pi_p=\pi_{\tilde p}$. A completely equivalent line of reasoning can be used for all parameters of the model, which implies the identification of each coefficient and concludes the proof of identifiability.

Regarding the combination of assumptions A1 to A3 with A6 and A7, let's assume that, for all $(y, \delta_T, \delta_C, x)$, the equality $l(\beta; y, \delta_T, \delta_C, x) = l(\tilde\beta; y, \delta_T, \delta_C, x)$ holds.
    
If $\delta_C=1$ we have, for all $(y, x)$,
\begin{equation}
    f_{\beta_{C}}(y|x)\left[1-h_{\tau}^{2}\left\{\pi_p(x)F_{\beta_{U}}(y|x)|F_{\beta_{C}}(y|x),x\right\}\right]= f_{\tilde\beta_{C}}(y|x)\left[1-h_{\tilde\tau}^{2}\left\{\pi_{\tilde p}(x)F_{\tilde\beta_{U}}(y|x)|F_{\tilde\beta_{C}}(y|x),x\right\}\right], \label{eq:deltaC1}
\end{equation}
and hence it follows from A2 that
$$\lim_{y\rightarrow 0} \frac{f_{\beta_{C}}(y|x)}{f_{\tilde\beta_{C}}(y|x)}=1 \hspace{0.1cm}\forall x \quad \text{or} \quad \lim_{y\rightarrow \infty} \frac{f_{\beta_{C}}(y|x)}{f_{\tilde\beta_{C}}(y|x)}=1\hspace{0.1cm}\forall x.$$
Thanks to assumption A1 we obtain $\beta_C=\tilde\beta_C$ and equation \eqref{eq:deltaC1} becomes 
\begin{equation}
	h_{\tau}^{2}\left\{\pi_p(x)F_{\beta_{U}}(y|x)|F_{\beta_{C}}(y|x),x\right\}=h_{\tilde\tau}^{2}\left\{\pi_{\tilde p}(x)F_{\tilde\beta_{U}}(y|x)|F_{\beta_{C}}(y|x),x\right\}. \label{eq:deltaC1bis}
\end{equation}

If $\delta_T=1$ we have, for all $(y, x)$,
\begin{align*}
    & \pi_p(x)f_{\beta_{U}}(y|x)\left[1-h_{\tau}^{1}\left\{F_{\beta_{C}}(y|x)|\pi_p(x)F_{\beta_{U}}(y|x),x\right\}\right] \nonumber \\
    & \quad\quad = \pi_{\tilde p}(x)f_{\tilde\beta_{U}}(y|x)\left[1-h_{\tilde\tau}^{1}\left\{F_{\beta_{C}}(y|x)|\pi_{\tilde p}(x)F_{\tilde\beta_{U}}(y|x),x\right\}\right], %\label{eq:deltaT1}
\end{align*}
it then follows from A3 that
$$ \lim_{y\rightarrow 0} \frac{\pi_p(x)f_{\beta_{U}}(y|x)}{\pi_{\tilde p}(x)f_{\tilde\beta_{U}}(y|x)}=1 \hspace{0.1cm}\forall x.$$
Thanks to assumption A6 we obtain $\beta_U=\tilde\beta_U$ and $p=\tilde p$. Hence, equation \eqref{eq:deltaC1bis} becomes, for all $(y, x)$, 
\begin{equation*}
	h_{\tau}^{2}\left\{\pi_p(x)F_{\beta_{U}}(y|x)|F_{\beta_{C}}(y|x),x\right\}=h_{\tilde\tau}^{2}\left\{\pi_{p}(x)F_{\beta_{U}}(y|x)|F_{\beta_{C}}(y|x),x\right\}
\end{equation*}
which, thanks to assumption A7, directly yields $\tau=\tilde\tau$ and concludes the proof.
\end{proof}

\begin{proof}[Proof of Theorem \ref{IdentA6A7}(a)]
We show that assumption A6 holds for the log-normal distribution, whose density, for $\sigma >0$ and $-\infty < \mu < \infty$, is given by
$$f_{\sigma,\mu}(y) = \frac{1}{y\sigma\sqrt{2\pi}} \cdot \exp\left\{-\frac{1}{2}\left(\frac{\log{(y)}-\mu}{\sigma}\right)^2\right\}.$$
Since the distribution parameters depend on the covariates, we write them as $\sigma(x)$ and $\mu(x)$.

The limit of interest is
\begin{align*}
	\lim_{y \rightarrow 0}\frac{\pi_p(x) f_{\beta_{U}}(y|x)}{\pi_{\tilde p}(x) f_{\tilde\beta_{U}}(y|x)} &= \lim_{y \rightarrow 0}\frac{\pi_p(x)\tilde\sigma(x)}{\pi_{\tilde p}(x)\sigma(x)} \cdot \exp\left\{-\frac{1}{2}\left(\frac{\log{(y)}-\mu(x)}{\sigma(x)}\right)^2+\frac{1}{2}\left(\frac{\log{(y)}-\tilde\mu(x)}{\tilde\sigma(x)}\right)^2\right\} \\
    &= \lim_{y \rightarrow 0}\frac{\pi_p(x)\tilde\sigma(x)}{
    \pi_{\tilde p}(x)\sigma(x)} \cdot \exp\Bigg\{\log{(y)}^2\left(\frac{1}{2\tilde\sigma(x)^2}-\frac{1}{2\sigma(x)^2}\right)-\log{(y)}\left(\frac{\tilde\mu(x)}{\tilde\sigma(x)^2}-\frac{\mu(x)}{\sigma(x)^2}\right) \\
    & \hspace{10cm} +\frac{1}{2}\left(\frac{\tilde\mu(x)^2}{\tilde\sigma(x)^2}-\frac{\mu(x)^2}{\sigma(x)^2}\right)\Bigg\}.
\end{align*}
Since $\displaystyle\lim_{y\rightarrow 0}\log{(y)}=-\infty$ and the dominant term in the exponent is $\log{(y)}^2\left(\left(1/2\tilde\sigma(x)^2\right)-\left(1/2\sigma(x)^2\right)\right)$, the exponential term converges either to $0$ or $\infty$, depending on the sign of $\left(\left(1/2\tilde\sigma(x)^2\right)-\left(1/2\sigma(x)^2\right)\right)$. Therefore, the above limit can equal $1$ only if $1/2\tilde\sigma(x)^2=1/2\sigma(x)^2$ and $\tilde\mu(x)/\tilde\sigma(x)^2=\mu(x)/\sigma(x)^2$. The first equality implies that $\sigma(x)=\tilde\sigma(x)$. Substituting this
relation into the second equality yields $\mu(x)=\tilde\mu(x)$. The exponential term is then equal to $1$, which directly implies that $\pi_p(x)=\pi_{\tilde p}(x)$, as required. By Lemma \ref{IdentAlpha}, this is equivalent to the identifiability of the coefficients, namely $\beta_{U}=\tilde\beta_{U}$ and $p=\tilde p$ since the link functions are one-to-one and the components of $X$ are linearly independent.

Similar arguments establish assumption A6 for the log-student-t distribution.
\end{proof}

\begin{proof}[Proof of Theorem \ref{IdentA6A7}(b)]
We establish assumption A6 for the non-truncated Weibull distribution. The proof for the non-truncated Gamma distribution is very similar, has been verified and can be provided, on request, by the contact author.

The Weibull density, with scale and shape parameters depending on the covariates, $\lambda(x) >0$ and $k(x) >0$, is given by

$$f_{\beta_{U}}(y|x) = \frac{k(x)}{\lambda(x)}\cdot\left(\frac{y}{\lambda(x)}\right)^{k(x)-1}\cdot\exp\left\{-\left(\frac{y}{\lambda(x)}\right)^{k(x)}\right\}.$$

The limit of interest is

\begin{align*}
	\lim_{y \rightarrow 0} \frac{\pi_p(x) f_{\beta_{U}}(y|x)}{\pi_{\tilde p}(x) f_{\tilde\beta_{U}}(y|x)} &= \lim_{y \rightarrow 0} \frac{\frac{\pi_p(x)k(x)}{\lambda(x)}\cdot\left(\frac{y}{\lambda(x)}\right)^{k(x)-1}\cdot\exp\left\{-\left(\frac{y}{\lambda(x)}\right)^{k(x)}\right\}}{\frac{\pi_{\tilde p}(x)\tilde k(x)}{\tilde\lambda(x)}\cdot\left(\frac{y}{\tilde\lambda(x)}\right)^{\tilde k(x)-1}\cdot\exp\left\{-\left(\frac{y}{\tilde\lambda(x)}\right)^{\tilde k(x)}\right\}} \\
	%&= \lim_{y \rightarrow 0} \frac{\pi_p(x)k(x)\tilde\lambda(x)}{\pi_{\tilde p}(x)\tilde k(x)\lambda(x)}\cdot\left(\frac{y}{\lambda(x)}\right)^{k(x)-1}\cdot\left(\frac{\tilde\lambda(x)}{y}\right)^{\tilde k(x)-1}\cdot\exp\left\{-\left(\frac{y}{\lambda(x)}\right)^{k(x)}+\left(\frac{y}{\tilde\lambda(x)}\right)^{\tilde k(x)}\right\} \\
	&= \lim_{y\rightarrow 0} \frac{\pi_p(x)k(x)\tilde\lambda(x)^{\tilde k(x)}}{\pi_{\tilde p}(x)\tilde k(x)\lambda(x)^{k(x)}}\cdot y^{k(x)-\tilde k(x)}\cdot\exp\left\{-\left(\frac{y}{\lambda(x)}\right)^{k(x)}+\left(\frac{y}{\tilde\lambda(x)}\right)^{\tilde k(x)}\right\}.
\end{align*}

This limit can equal $1$ only if $k(x)=\tilde k(x)$. Moreover, since $\displaystyle\lim_{z\rightarrow 0}\exp{(z)}=1$, it is further required that $\big(\pi_p(x)\tilde\lambda(x)^{k(x)}\big)/\big(\pi_{\tilde p}(x)\lambda(x)^{k(x)}\big)=1$. By hypothesis, this equality holding for all $x$ is equivalent to $\beta_{U}=\tilde\beta_{U} \text{ and } p=\tilde p$. Therefore, assumption A6 is satisfied.
\end{proof}

\begin{proof}[Proof of Theorem \ref{IdentA6A7}(c)]
First of all, it should be noted that we will hereafter use the usual copula dependence parameter $\theta$ which is related to Kendall's tau by the following formula:
$$\tau^{K}(x) = 4 \int_0^1\int_0^1 \C_{\theta}(u, v|x) c_{\theta}(u, v|x) dudv -1,$$
where $c_{\theta}(u, v|x)$ is the copula density. Furthermore, all well-known copula families admit relationships between $\tau^{K}$ and $\theta$, allowing estimation of one from the other; see \cite{Nelsen} for some examples.
%An advantage of using Kendall’s tau instead of the native copula parameter $\theta$ is that the former is more readily interpretable. The domain of $\tau$ is [-1, 1], whereas $\theta$ belongs to a copula-specific domain. This parametrisation also simplifies calibration and simulation procedures.

Let us prove assumption A7 for the Frank copula. Proofs for each of the other copulas mentioned in the theorem have been checked by the contact author and can be provided on request to any interested reader.

Since the dependency parameter $\theta$ is in bijection with $\tau^{K}$, it is also impacted by the covariates and is denoted $\theta(x)$. For $\theta(x)\in(-\infty,\infty)\backslash\{0\}$,
$$\C_{\tau}(u, v|x)=-\frac{1}{\theta(x)}\log\left( 1+\frac{\left(e^{-\theta(x) u}-1\right)\left(e^{-\theta(x) v}-1\right)}{e^{-\theta(x)}-1} \right),$$
and
$$h_{\tau}^{2}\left(u|v,x\right)=\frac{\partial}{\partial v}\C_{\tau}(u, v|x)=\frac{e^{-\theta(x) v}\left( e^{-\theta(x) u}-1 \right)}{\left( e^{-\theta(x)}-1 \right)+\left( e^{-\theta(x) u}-1 \right)\left( e^{-\theta(x) v}-1 \right)}.$$
The equality $h_{\tau}^{2}\left\{\pi_p(x)F_{\beta_{U}}(y|x)|F_{\beta_{C}}(y|x),x\right\}=h_{\tilde\tau}^{2}\left\{\pi_{p}(x)F_{\beta_{U}}(y|x)|F_{\beta_{C}}(y|x),x\right\}$ is then
\begin{align*}
    & \frac{e^{-\theta(x) F_{\beta_{C}}(y|x)}\left( e^{-\theta(x)\pi_p(x)F_{\beta_{U}}(y|x)}-1 \right)}{\left( e^{-\theta(x)}-1 \right)+\left( e^{-\theta(x)\pi_p(x)F_{\beta_{U}}(y|x)}-1 \right)\left( e^{-\theta(x) F_{\beta_{C}}(y|x)}-1 \right)} \nonumber \\
    & \hspace{5cm} = \frac{e^{-\tilde\theta(x) F_{\beta_{C}}(y|x)}\left( e^{-\tilde\theta(x)\pi_p(x)F_{\beta_{U}}(y|x)}-1 \right)}{\left( e^{-\tilde\theta(x)}-1 \right)+\left( e^{-\tilde\theta(x)\pi_p(x)F_{\beta_{U}}(y|x)}-1 \right)\left( e^{-\tilde\theta(x) F_{\beta_{C}}(y|x)}-1 \right)}.
\end{align*}
Taking the limit when $y$ tends to $\infty$ leads to
$$\frac{e^{-\theta(x)}\left( e^{-\theta(x)\pi_p(x)}-1 \right)}{\left( e^{-\theta(x)}-1 \right)+\left( e^{-\theta(x)\pi_p(x)}-1 \right)\left( e^{-\theta(x)}-1 \right)}=\frac{e^{-\tilde\theta(x)}\left( e^{-\tilde\theta(x)\pi_p(x)}-1 \right)}{\left( e^{-\tilde\theta(x)}-1 \right)+\left( e^{-\tilde\theta(x)\pi_p(x)}-1 \right)\left( e^{-\tilde\theta(x)}-1 \right)}$$
which is equivalent to
$$\frac{e^{-\theta(x)}\left( e^{-\theta(x)\pi_p(x)}-1 \right)}{e^{-\theta(x)\pi_p(x)}\left( e^{-\theta(x)}-1 \right)}=\frac{e^{-\tilde\theta(x)}\left( e^{-\tilde\theta(x)\pi_p(x)}-1 \right)}{e^{-\tilde\theta(x)\pi_p(x)}\left( e^{-\tilde\theta(x)}-1 \right)} \Leftrightarrow \frac{1- e^{\theta(x)\pi_p(x)}}{1-e^{\theta(x)}}=\frac{1-e^{\tilde\theta(x)\pi_p(x)}}{1-e^{\tilde\theta(x)}}.$$
Since this equality need to be verified for all $x$, and that for every fixed $x$ the function $g(z) = (1-e^{z\pi_p(x)})/(1-e^{z})$ is strictly decreasing, we deduce that $\theta(x)$ must be equal to $\tilde\theta(x)$ for all $x$. The bijection between $\theta(x)$ and $\tau^{K}(x)$ implies that $\tau^{K}_{\tau}(x)=\tau^{K}_{\tilde\tau}(x)$ for all $x$. Applying Lemma \ref{IdentAlpha} yields $\tau=\tilde\tau$ as requested.
\end{proof}

Finally, here is the proof of Remark 3.1, which allows to verify the condition stated in Theorem \ref{IdentA6A7}(b).

\begin{proof}[Proof of Remark \ref{RmkIdent}]
%Since the direction $\beta_{U}=\tilde\beta_{U} \text{ and } p=\tilde p \Longrightarrow \frac{\pi_p(x)\tilde\lambda(x)^{k(x)}}{\pi_{\tilde p}(x)\lambda(x)^{k(x)}}=1 \quad\forall x$ is obvious, the point is to prove that $\frac{\pi_p(x)\tilde\lambda(x)^{k(x)}}{\pi_{\tilde p}(x)\lambda(x)^{k(x)}}=1 \quad\forall x$ implies $\beta_{U}=\tilde\beta_{U} \text{ and } p=\tilde p$.

%Lets' assume that ...

%--------------------------

 %Let’s now check that this equation holds true if the shape parameter does not take the form $$k(x)=\frac{\log\left(\frac{\pi_{\tilde p}(x)}{\pi_p(x)}\right)}{\log\left(\frac{\tilde\lambda(x)}{\lambda(x)}\right)} \quad\text{for any } x.$$
The following evidence pertains to the case of a continuous covariate. When the covariate is binary, the model is not identifiable in this scenario.

We have the following equivalences
\begin{align*}
	\frac{\pi_p(x)\tilde\lambda(x)^{k(x)}}{\pi_{\tilde p}(x)\lambda(x)^{k(x)}}=1 \quad\forall x &\Longleftrightarrow \left(\frac{\tilde\lambda(x)}{\lambda(x)}\right)^{k(x)} = \frac{\pi_{\tilde p}(x)}{\pi_p(x)} \quad\forall x \\
	&\Longleftrightarrow k(x)\log\left(\frac{\tilde\lambda(x)}{\lambda(x)}\right) = \log\left(\frac{\pi_{\tilde p}(x)}{\pi_p(x)}\right) \quad\forall x. %\\
	%&\Longleftrightarrow \left\{\begin{aligned}
%\log\left(\frac{\tilde\lambda(x)}{\lambda(x)}\right) &= 0 %\quad\text{for }x\in X_A\\
%\log\left(\frac{\tilde\lambda(x)}{\lambda(x)}\right) &\neq 0 %\quad\text{for }x\in X_B
%\end{aligned}\right. \quad\text{with }X_A\cup X_B=[-1,1]\\
%&\Longleftrightarrow \left\{\begin{aligned}
%\tilde\lambda(x) &= \lambda(x) \quad\text{for }x\in X_A\\
%k(x) &= \frac{\log\left(\frac{\pi_{\tilde p}(x)}{\pi_p(x)}\right)}{\log\left(\frac{\tilde\lambda(x)}{\lambda(x)}\right)} \quad\text{for }x\in X_B
%\end{aligned}\right. \quad\text{with }X_A\cup X_B=[-1,1].
\end{align*}
Depending on the values taken by the covariates, there are only two possibilities. Either
$$\log\left(\frac{\tilde\lambda(x)}{\lambda(x)}\right) = 0 \Longleftrightarrow \tilde\lambda(x) = \lambda(x),$$
or
$$\log\left(\frac{\tilde\lambda(x)}{\lambda(x)}\right) \neq 0 \Longleftrightarrow k(x) = \frac{\log\left(\frac{\pi_{\tilde p}(x)}{\pi_p(x)}\right)}{\log\left(\frac{\tilde\lambda(x)}{\lambda(x)}\right)}.$$
Assuming that
$ k(x) \neq \frac{\log\left(\frac{\pi_{\tilde p}(x)}{\pi_p(x)}\right)}{\log\left(\frac{\tilde\lambda(x)}{\lambda(x)}\right)} \quad\forall x$
is then equivalent to say that $\tilde\lambda(x)=\lambda(x)$ for all $x$. The initial equality becomes
$$\frac{\pi_p(x)\lambda(x)^{k(x)}}{\pi_{\tilde p}(x)\lambda(x)^{k(x)}}=1 \quad\forall x \quad\Longleftrightarrow\quad \pi_p(x)=\pi_{\tilde p}(x) \quad\forall x.$$
Applying lemma \ref{IdentAlpha} directly leads to the identification of each coefficients, $\beta_{U}=\tilde\beta_{U}$ and $p=\tilde p$, since the link functions are one-to-one and the components of $X$ are linearly independent.

\end{proof}

\end{document}